%% file: document.tex
\documentclass{article}
\usepackage{amsthm}
\usepackage{amsmath}
\usepackage{amssymb}
\usepackage{cite}
\usepackage{hyperref}
\usepackage{color}
\usepackage{xspace}
\usepackage{authblk}
\usepackage[noend]{algpseudocode}
\usepackage{algorithm}
\usepackage{enumerate}

\newtheorem*{remark*}{Remark}
\newtheorem{example}{Example}
\newtheorem{theorem}{Theorem}
\newtheorem{lemma}[theorem]{Lemma}
\newtheorem{proposition}[theorem]{Proposition}

\newtheorem{corollary}[theorem]{Corollary}
\newtheorem{definition}[theorem]{Definition}

\input{macros}

\begin{document}

\title{Lasserre Lower Bounds and Definability of Semidefinite Programming}
\author{Anuj Dawar and Pengming Wang}
\affil{University of Cambridge Computer Laboratory\\
	\texttt{\{anuj.dawar, pengming.wang\}@cl.cam.ac.uk}}

\maketitle

\begin{abstract}
For a large class of optimization problems, namely those that can be expressed
as finite-valued constraint satisfaction problems (VCSPs), 
we establish a dichotomy on the number of levels of the 
Lasserre hierarchy of semi-definite programs (SDPs) that are 
required to solve the problem exactly.  In particular, we show that if a 
finite-valued constraint problem is not solved exactly by its basic linear 
programming relaxation, it is also not solved exactly by any sub-linear 
number of levels of the Lasserre hierarchy.
The lower bounds are established through
logical undefinability results. 
We show that the linear programming relaxation of the problem, 
as well as the SDP corresponding to any fixed level of the Lasserre hierarchy 
is interpretable in a VCSP instance by means of formulas of
fixed-point logic with counting.  
We also show that the solution of an SDP can be expressed in this logic.  
Together, these results give a way of translating lower bounds on the number of
variables required in counting logic to express a 
VCSP into lower bounds on the number of levels 
required in the Lasserre hierarchy to eliminate the integrality gap.

As a special case, we obtain the same dichotomy for the class of $\MAXCSP$ problems, 
generalizing some earlier Lasserre lower bound results from~\cite{schoenebeck08}.

\end{abstract}

\section{Introduction}\label{sec:intro}
\input{intro}

\section{Background}\label{sec:background}
\input{background}

\section{Main result}\label{sec:result}

\input{result}

\section{Counting width of finite-valued \CSP s}\label{sec:linbound}
\input{linbound}

\section{Expressing semidefinite programs}\label{sec:sdp}
\input{sdp}

\section{Lasserre lower bounds}\label{sec:lasserre}
\input{lasserre}

\section{Conclusion}\label{sec:conclusion}
\input{conclusion}

\bibliography{ref}

\end{document}

%% file: macros.tex
\newcommand{\QQ}{\ensuremath{\mathbb{Q}}}
\newcommand{\NN}{\ensuremath{\mathbb{N}}}
\newcommand{\ZZ}{\ensuremath{\mathbb{Z}}}

\newcommand{\NP}{\ensuremath{\mathrm{NP}}\xspace}

\newcommand{\Ck}{\ensuremath{\mathrm{C}^k}\xspace}
\newcommand{\VCSP}{\ensuremath{\mathrm{VCSP}\xspace}}
\newcommand{\MAXCSP}{\ensuremath{\mathrm{MAXCSP}\xspace}}
\newcommand{\CSP}{\ensuremath{\mathrm{CSP}\xspace}}
\newcommand{\SDP}{\ensuremath{\mathrm{SDP}\xspace}}
\newcommand{\BLP}{\ensuremath{\mathrm{BLP}\xspace}}

\newcommand{\LP}{\ensuremath{\mathrm{LP}\xspace}}
\newcommand{\BIT}{\ensuremath{\mathrm{BIT}\xspace}}

\newcommand{\MAXCUT}{\ensuremath{\mathrm{MAXCUT}\xspace}}
\newcommand{\SAT}{\ensuremath{\mathrm{SAT}\xspace}}
\newcommand{\Las}{\ensuremath{\mathrm{Las}\xspace}}
\newcommand{\conv}{\ensuremath{\mathrm{conv}\xspace}}

\newcommand{\XORSAT}{\ensuremath{\mathrm{XORSAT}}\xspace}

\newcommand{\dunion}{\ \dot\cup\ }

\newcommand{\logic}[1]{\text{\upshape #1}\xspace}
\newcommand{\FPC}{\logic{FPC}}

\newcommand{\ar}{\ensuremath{\mathrm{ar}}}

\newcommand{\logicoperator}[1]{\mathbf{#1}}
\newcommand{\ifpop}{\logicoperator{ifp}}

\renewcommand{\vec}[1]{\bar{#1}}
\newcommand{\tup}[1]{\vec{#1}}
\newcommand{\countingTerm}[1]{\#_{#1}}

\newcommand{\struct}[1]{\ensuremath{\mathbf #1}\xspace}
\newcommand{\univ}[1]{\ensuremath{\dom(\struct #1)}} 
\newcommand{\innerprod}[2]{\langle #1, #2\rangle} 
\newcommand{\powerset}{\wp}
\newcommand{\norm}[1]{\| #1 \|}

\newcommand{\dom}{\mathrm{dom}}

\newcommand{\Opt}{\mathrm{Opt}}
\newcommand{\Val}{\mathrm{Val}}
\newcommand{\vocvec}{\tau_{\text{vec}}}
\newcommand{\vocmat}{\tau_{\text{mat}}}
\newcommand{\LIN}{\mathrm{LIN}}

\newcommand{\defeq}{:=}

\newcommand{\argmax}{\mathrm{argmax}}

\newcommand{\tth}{\ensuremath{t^{\scriptsize\mbox{th}}}}

%% file: intro.tex
Many natural optimization problems can be expressed as 0--1 integer programming 
problems.  Indeed, since the problem of determining an optimal solution to a 0--1
integer programming problem is $\NP$-complete, in principle any problem in $\NP$
can be so expressed in this framework.  Any integer programming problem admits a
linear programming relaxation obtained by dropping the integrality constraints. 
This relaxed linear program can then be solved by standard polynomial-time 
algorithms, however it admits solutions that are not solutions to the original
integer program.  The gap between the optimal solution to the integer program 
and its linear programming relaxation is known as an \emph{integrality gap}.  
There are various ways that the linear programming relaxation may be tightened 
by additional constraints to more closely correspond to the original problem.  
Several systematic ways have been studied in the literature of constructing 
hierarchies of ever tighter linear or semidefinite programs, 
including those of Sherali-Adams~\cite{SA90}, Lovasz-Schrijver~\cite{LS91} 
and Lasserre~\cite{Lasserre:2001kq}.  Of these, the Lasserre hierarchy is the strongest and gives, for each $t$, a semidefinite program of size $n^{O(t)}$ (where $n$ 
is the size of the original integer program) that defines a feasible region 
whose projection on to the original variables includes the solutions of the 
integer program.  When $t = n$, this projection is exactly the convex hull of 
the solutions to the original integer program.  When this can be achieved for smaller values of 
$t$, we get substantially faster algorithms for solving (possibly 
approximately) the original problem.  For many combinatorial optimization 
problems, the Lasserre relaxations provide the best known approximation 
algorithms (see~\cite{Chlamtac2012}).

Our aim in this paper is to establish integrality lower bounds, 
i.e.\ to establish for particular combinatorial optimization problems $P$, 
expressed as an integer programming  problem, a lower bound on the value of $t$ 
such that the $t$-th level of the Lasserre hierarchy yields the convex hull of 
the feasible region of $P$.  Schoenebeck~\cite{schoenebeck08} 
established such lower bounds for the Lasserre hierarchy, showing a linear lower 
bound on $t$ for a variety of Boolean constraint satisfaction problems, 
including $\text{Max-}k\text{-}\XORSAT$.  
We show that those lower bounds are part of a general pattern.  Indeed, we 
demonstrate a dichotomy on the minimum value of $t$ needed in the Lasserre
hierarchy to establish exact solutions to optimization problems in the framework
of \emph{finite-valued constraint satisfaction problems} ($\VCSP$s). 
That is, any such $\VCSP$ is either already solved by simply relaxing the integrality
constraints in its 0--1 linear program formulation (resulting in the 
so-called \emph{basic linear program relaxation} ($\BLP$)),
or requires a linear number of Lasserre relaxation steps to be solved exactly.
As a direct consequence, we obtain the same dichotomy for the class of (weighted) 
$\MAXCSP$ problems.

The study of the complexity of $\VCSP$s has been quite successful
in the recent past, culminating in the dichotomy result by Thapper and \v{Z}ivn\'y~\cite{tz13:stoc}.
There, a complete characterization of the tractable and intractable cases of $\VCSP$s
are shown. Namely, any $\VCSP$ is either solved exactly by its $\BLP$;
or the problem $\MAXCUT$ reduces to it, and it is $\NP$-hard. Our main result
complements this dichotomy by showing a linear lower bound for the levels of Lasserre
relaxations required to exactly solve the hard cases. 
This characterization of $\VCSP$s into those solvable by their $\BLP$,
and those to which $\MAXCUT$ reduces, also applies in the context of logical definability.
It is known from~\cite{DW15} that in the former case, the class is definable in
fixed-point logic with counting (\FPC), while in the latter case there is no constant $k$ 
such that it is definable using only $k$ variables, even in an infinitary 
logic with counting.

In the present paper we establish a more fine-grained view on the above undefinability
result by lifting a previous undefinability result on classical $\CSP$s 
from~\cite{Atserias20091666}, and connect it to the number of levels of 
the Lasserre hierarchy needed to exactly capture the feasible region of a $\VCSP$.  
To be precise, we show that if $k$ variables are required to define the $\VCSP$ in 
logic with counting, then $\Omega(k)$ levels of the Lasserre hierarchy are 
needed to capture the corresponding feasible region.
Our result is established in two significant steps.  On the one hand, for each 
$k$, there is an $\FPC$ interpretation that constructs from an integer program 
the $k$-th Lasserre relaxation of that program by extending methods 
of~\cite{DW15}.  On the other hand, we show, using methods 
of~\cite{adh13:lics,adh15:jacm}, that there is a an $\FPC$ 
interpretation that can, given an explicitly given semidefinite program, 
define its optimal solution, up to a given approximation.
The dichotomy is then completed by showing that, 
for every $\VCSP$ that is not captured by the basic linear program relaxation, 
there is a linear lower bound on the number of variables required to define it.

We begin by introducing the necessary definitions from combinatorial 
optimization, logic and constraint satisfaction problems in 
Section~\ref{sec:background}.  Section~\ref{sec:result} formulates the main 
result and the steps establishing are then given in Sections~\ref{sec:linbound},~\ref{sec:sdp},
and~\ref{sec:lasserre}.

%% file: background.tex
\textbf{Notation.}  
We write $\NN$ for the natural numbers, $\ZZ$ for the integers,
$\QQ$ for the rational numbers, and use a superscript plus to denote the 
non-negative subset, e.g.\ $\ZZ^+=\NN\cup\{0\}$.

Given sets $A$ and $I$,
an $A$-valued \emph{vector} $v$ indexed by $I$ is a function $v:I\rightarrow A$. Often we
simply use the subscript notation, writing $v_i$ for $v(i)$. When there is no explicit index set given, 
vectors are indexed by an initial segment of the natural numbers: $\{1,\dots,d\}$.
A \emph{matrix} $M$ is a vector which is indexed by a product set $I\times J$. We use
$M_{i,j}$ to denote $M(i,j)$. We write $M^T$ to denote the \emph{transposed} matrix of $M$, 
defined as $M^T_{i,j}:=M_{j,i}$. A matrix is called \emph{symmetric} if $M=M^T$.

For a rational valued vector $v\in\QQ^I$ over some index set $I$, its \emph{norm}
$\norm{v}$ is defined as the $L^2$-norm over $\QQ^I$, that is, 
$\norm{v}:=\sqrt{\sum_{i\in I}v_i^2}$. The \emph{inner product} of two
vectors $a,b\in\QQ^I$ is defined as $\innerprod{a}{b}:=\sum_{i\in I}a_ib_i$. We
also occasionally write this as a matrix multiplication of vectors $a^T b$. In
the case of matrices, this definition of $\innerprod{\cdot}{\cdot}$ coincides 
with the trace inner product
of matrices. Note that the norm of a matrix $M$ is defined as the norm of $M$
seen as a vector. For a set $\mathcal{F}\subseteq \QQ^I$ and a vector $v\in\QQ^I$,
we define the \emph{distance} $d(v,\mathcal{F})$ in the usual way, i.e.\ as 
$d(v,\mathcal{F}):=\min_{x\in\mathcal{F}}\norm{x-v}$. The \emph{ball} $\mathcal{B}(v,r)$
around $v$ with radius $r\in\QQ$ is then the set $\mathcal{B}(v,r):=\{x\in\QQ^I\mid \norm{x-v}\leq r\}$.

A rational valued symmetric matrix $M\in\QQ^{I\times I}$ is \emph{positive
semidefinite}, if for any $x\in\QQ^I$, it holds $x^TMx\geq 0$.
We use $M\succeq 0$ to denote that $M$ is positive semidefinite.

We use the notion of \emph{polynomial time} solvable as used in \cite{GLS88}. That is,
for problems where we expect an exact solution, this means that
there exists an algorithm running in time
polynomial to the encoding of the input that returns an exact solution to the
problem. For a problem with a given error parameter $\delta>0$, we say it is 
polynomial time solvable if there is an algorithm running in time polynomial
of the encoding of the input and $\log(1/\delta)$, that solves the problem up to the 
specified error.

\subsection{Constraint Satisfaction Programs}\label{sec:csps}
The class of constraint satisfaction problems (CSPs) provides a framework in which many
common combinatorial problems can be expressed. Examples include $k$-colouring,
$k$-satisfiability, solvability of linear equations over a finite field, 
and many more. Here we consider CSPs that are parameterized by a fixed \emph{domain}
and a \emph{constraint language}, and their optimization variant of so called 
\emph{finite-valued} CSPs. Finite-valued CSPs have been extensively studied in the
recent past and generalize common optimization problems, such as the class of
$\MAXCSP$s. 

\begin{definition}
	A \emph{domain} $D$ is a finite, non-empty set. A
	\emph{constraint language} $\Gamma$ over $D$ is a set of relations over $D$, where
	each $R\in\Gamma$ is a relation of some arity $m=\ar(R)$, and $R\subseteq D^m$.
	
	An instance of the \emph{constraint satisfaction problem} over $(D,\Gamma)$
	is then a pair $I=(V,C)$, where $V$ is a finite set of \emph{variables}, and $C$
	is a finite set of \emph{constraints}. Each constraint $c\in C$ is a
	pair $(s,R)$ associating a relation $R\in\Gamma$ with a \emph{scope} $s\in V^{\ar(R)}$.
	
	We say an assignment $h:V\rightarrow D$ 
	\emph{satisfies} a constraint $c=(s,R)$ if $h(s)\in R$. The goal is to decide
	whether there exists an assignment that satisfies all constraints $c\in C$.
\end{definition}

\begin{example}
	Let $D=\{0,1,2\}$, and $\Gamma=\{\neq\}$, that is, $\Gamma$ contains the 
	inequality relation over $D$. We obtain the
	$3$-colouring problem as $\CSP(D,\Gamma)$: A instance graph $G=(V,E)$ is
	simply interpreted as a $\CSP(D,\Gamma)$-instance where the variables are the 
	vertices $V$, and every edge $(u,v)\in E$ induces the constraint $((u,v),\neq)$.
\end{example} 

Instead of just deciding whether an instance is satisfiable or not, we could be
interested in \emph{how many} constraints could be satisfied at the same time.
This is the optimization problem $\MAXCSP$: We let $\MAXCSP(D,\Gamma)$ be the 
optimization problem of determining the maximal number of constraints that can 
be satisfied in a given instance of $\CSP(D,\Gamma)$.

This can be further generalized. In $\MAXCSP$s, every constraint itself is either
satisfied or not. If we now allow constraints to be satisfied to 
\emph{different degrees}, we obtain the framework of \emph{finite-valued CSPs} 
($\VCSP$s). This is the framework we will be working in.

Here, a constraint language consists of a set of functions, instead
of relations, where a $m$-ary function $f:D^m\rightarrow \ZZ^+$ assigns each 
$m$-tuple of the domain an integer value. A
constraint then associates some tuple of variables with such a function, and we 
are interested in the maximum value that can be achieved by any assignment. This is
formalized in the following definition.

\begin{definition}
	Let $D$ be a domain. A \emph{finite-valued constraint language} $\Gamma$ 
	is a set of functions,
	where each $f\in\Gamma$ has some arity $m=\ar(f)$, and $f:D^m\rightarrow\ZZ^+$.
	
	An instance of the \emph{valued constraint satisfaction problem} over $(D,\Gamma)$
	is a pair $I=(V,C)$, where $V$ is a finite set of \emph{variables}, and $C$
	is a finite set of \emph{constraints}. Each constraint $c\in C$ is a
	triple $(s,f,w)$ associating a relation $f\in\Gamma$ with a \emph{scope} $s\in V^{\ar(f)}$
	and a \emph{weight} $w\in\ZZ^+$.
	
	The \emph{value} of an assignment $h:V\rightarrow D$ for an instance $I$
	is given as $\Val_I(h):=\sum_{(s, f, w)\in C} w\cdot f(h(s))$. The goal is
	to determine the maximum value $\Opt(I)=\max_h \Val_I(h)$.
\end{definition}

\begin{example}
	Let $\MAXCUT$ be the problem of determining the value of the
	maximum cut in a graph $G=(V,E)$ with edge weights
	$w:E\rightarrow \ZZ^+$.
	Furthermore, we fix $D=\{0,1\}$, and $\Gamma=\{f\}$ 
	where $f(x,y)=0$ if $x=y$, and $f(x,y)=1$
	if $x\neq y$.  An instance of $\MAXCUT$ can be interpreted as an 
	instance $I$ of $\VCSP(D,\Gamma)$:
	The variables are the vertices $V$, and we have the constraint $((u,v),f,w(u,v))$
	for every edge $(u,v)\in E$. The value of the maximum cut is exactly $\Opt(I)$.
\end{example} 

\begin{example}
	It is not difficult to see that every $\MAXCSP(D,\Gamma)$ is a finite-valued CSP.
	For any $m$-ary relation $R$, define a function
	$f_R:D^m\rightarrow\{0,1\}$ as $f(t)=1$ if $t\in R$ and $f(t)=0$ if $t\notin R$.
	Let $\Gamma'$ be the finite-valued constraint language that consist of 
	$f_R$ for all $R\in\Gamma$. Then, $\MAXCSP(D,\Gamma)=\VCSP(D,\Gamma')$.
\end{example}

When talking complexity classes and reductions, it is often more 
convenient to also phrase VCSPs as decision 
problems. Abusing notation, we will use $\VCSP(D,\Gamma)$ to also denote the set
of pairs $(I,t)$ such that $I$ is an instance with $\Opt(I)\geq t$. The decision
problem is then, given any pair $(I,t)$, to decide whether $(I,t)\in\VCSP(D,\Gamma)$.

In the study of valued constraint satisfaction problems, linear programming in
particular has proven to be a useful tool. In fact, every instance of $\VCSP(D,\Gamma)$
is equivalent to the following \emph{integer} linear program.

For an instance $I=(V,C)$, the program contains variables $\lambda_{c,x}$
for every $c\in C$ with $c=(s,f,w)$ and $x\in D^{\ar(s)}$,
and $\mu_{v,a}$ for every $v\in V$ and $a\in D$. A solution that sets
a variable $\lambda_{c,x}$ to $1$ then corresponds to an assignment that
assigns the scope of the constraint $c$ to the tuple $x$. In order to maintain 
consistency of the assignment between constraints, the variable $\mu_{v,a}$
encodes whether the variable $v$ is assigned the value $a$. The objective is then
to maximize the value of the assignment. The 0--1
program is then given below.
\begin{align*}
	 \max  \sum_{c\in C}\sum_{x\in D^{\ar(s)}}&\lambda_{c,x}\cdot w\cdot f(x) \quad \text{where } c=(s,f,w) \text{, s.t. }\\
	\sum_{x\in D^{\ar(s)}; x_i=a} &\lambda_{c,x} = \mu_{s_i,a} &&\forall c\in C, a\in D, i\in [\ar(s)] \\
	\sum_{a\in D} &\mu_{v,a} = 1 &&\forall v\in V \\
	& \lambda_{c,x} \in \{0,1\} &&\forall c\in C, x\in D^{\ar(s)}\\
	&\mu_{v,a} \in \{0,1\} &&\forall v\in V, a\in D
\end{align*}

If we relax the integrality constraints of the above LP to 
$0\leq \lambda_{c,x}\leq 1$ and $0\leq \mu_{v,a}\leq 1$, 
we obtain the \emph{basic linear program relaxation} $\BLP(I)$. Since this allows
rational assignments, this LP can be solved exactly in polynomial time. In general
the optimal value of $\BLP(I)$ only gives an overestimate of the optimal value $\Opt(I)$
to the $\VCSP$. However, there are $(D,\Gamma)$ for which any instance $I$ of 
$\VCSP(D,\Gamma)$ is solved by
$\BLP(I)$ exactly, and solving the LP gives an exact algorithm 
for $\VCSP(D,\Gamma)$. Thapper and \v{Z}ivn\'y\cite{tz13:stoc} give a
complete characterization of those cases -- and show that in all other cases
the problem is $\NP$-hard.

\begin{theorem}\label{thm:vcsp}
	For any domain $D$, and any finite-valued constraint language $\Gamma$, 
	either every instance $I$ of $\VCSP(D,\Gamma)$ is solved by $\BLP(I)$; or
	the problem $\MAXCUT$ polynomial-time reduces to $\VCSP(D,\Gamma)$.
\end{theorem}

Our main result expands on this dichotomy result, and shows a linear lower bound of
the required levels of the Lasserre hierarchy for all the cases not solved by
the BLP relaxation.

\subsection{Semidefinite Optimization}
We give a brief overview of the basic notions of semidefinite programs.

In general, semidefinite programming refers to a framework of constrained 
optimization problems where the search space is over the set of 
\emph{positive semidefinite matrices}. More specifically, in a typical 
semidefinite program we are interested in the entries of a 
symmetric matrix $X\in\QQ^{V\times V}$ that maximizes the value of an objective function $\innerprod{C}{X}$,
subject to a set of constraints of the form
$\innerprod{A_i}{X}\leq b_i$ 
with the additional constraint that $X$ is a positive semidefinite matrix. 

\begin{definition}
	Let $V,M$ be sets, and let $V$ be non-empty.
	A \emph{semidefinite program (SDP)} is given by
	an objective matrix $C\in\QQ^{V\times V}$, a $\QQ^{V\times V}$-valued vector  
	$\mathcal{A}\in \QQ^{M\times (V\times V)}$,
	and a vector $b\in\QQ^M$.
	
	We call $\mathcal{F}_{\mathcal{A},b}:=\{X\in\QQ^{V\times V}\mid X\succeq 0, \innerprod{A_i}{X}\leq b_i, A_i=\mathcal{A}(i), i\in M\}$
	the set of \emph{feasible solutions}. 
\end{definition}

We sometimes call sets that can be defined as feasible regions
of an SDP a \emph{positive semidefinite set}. This definition covers SDPs that are
in the so-called \emph{conic standard form}. Sometimes however it is more convenient
to specify SDPs in their \emph{inequality standard form}.
In this form, the SDP is instead given by a matrix
$Z\in\QQ^{M\times M}$, a matrix-valued vector $\mathcal{Y}\in\QQ^{V\times(M\times M)}$,
and an objective vector $c\in\QQ^V$.
The feasible region is then defined as
$\mathcal{F}_{\mathcal{Y},Z}:=\{x\in\QQ^V\mid Z+\sum_{v\in V}x_v \mathcal{Y}_v \succeq 0\}$.
The two standard forms can be converted into each other by adding,
substituting, and rearranging variables.   The number of additional variables needed can be bounded by a linear function in both cases.  Hence, we will use whichever representation is most convenient for any given case.

Note that by the definition of inner product, the objective 
function $\innerprod{C}{X}$ is a linear function over the variables $x_{u,v}, u,v\in V$.
Likewise, constraints of the form $\innerprod{A_i}{X}\leq b_i$ are also linear inequalities over the entries $x_{u,v}$. Hence, we can view
semidefinite programs as a generalization of \emph{linear programs},
with the additional constraint that the solution must define a positive 
semidefinite matrix. In fact, the semidefinite constraint $X\succeq 0$ 
essentially imposes an \emph{infinite} set of additional linear constraints,
namely $a^TXa\geq 0$ for all $a\in\QQ^V$.

The feasible region of a semidefinite program is convex, since 
it can be described as an intersection of (infinitely many) halfspaces.
Classically, in the context of convex optimization, we are interested in the
solutions of the two main problems of \emph{optimization} and \emph{separation}.
As a technical point, the optimal solution to a semidefinite program, or convex
problems in general, is not necessarily rational, so we can only express it
up to a finite precision. This gives rise to the weak formulations of the problems
where we allow an additive error to be specified in the input.

\begin{definition}
	Let $V$ be a non-empty set. Given a vector $c\in\QQ^V$ and a convex set
	$\mathcal{F}\subseteq \QQ^V$, the \emph{strong optimization problem}
	is to either find an element $y=\argmax_{x \in\mathcal{F}}\innerprod{c}{x}$,
	or to determine that $\mathcal{F}$ is empty, or that 
	$\max_{x\in\mathcal{F}}\innerprod{c}{x}$ is unbounded.
	
	In the \emph{weak optimization problem} we are given an additional
	error parameter $\delta>0$, and want to determine an element $y$ that is 
	$\delta$-close to $\mathcal{F}$, i.e.\ $d(y,\mathcal{F})\leq\delta$,
	that is also $\delta$-maximal, i.e.\ 
	$\innerprod{c}{y}+\delta\geq \max_{x\in\mathcal{F}}\innerprod{c}{x}$,
	or, again, to determine that $\max_{x\in\mathcal{F}}\innerprod{c}{x}$ is unbounded.
\end{definition}

\begin{definition}
	Let $V$ be a non-empty set. Given a vector $y\in\QQ^V$ and a convex set
	$\mathcal{F}\subseteq \QQ^V$, the \emph{strong separation problem}
	is the problem of determining either 
	that $y\in\mathcal{F}$, or finding a vector $s\in\QQ^V$ with
	$\innerprod{s}{y}>\max\{\innerprod{s}{x}\mid x\in\mathcal{F}\}$ and 
	$\norm{s}_\infty=1$. 
	
	In the \emph{weak separation problem}, we are given an additional parameter
	$\delta>0$, and are looking to determine that either $y$ is
	$\delta$-close to $\mathcal{F}$, i.e.\ $d(y,\mathcal{F})\leq\delta$,
	or to find a vector $s\in\QQ^V$, such that 
	$\innerprod{s}{y}+\delta>\max\{\innerprod{s}{x}\mid x\in\mathcal{F}\}$ and 
	$\norm{s}_\infty=1$.
\end{definition}

The relationship between the optimization and separation problem of a given convex
set is well-studied and is most prominently expressed by Gr\"otschel, Lov\'asz, and
Shrijver \cite{GLS88} as being polynomial time equivalent. More
precisely, with the additional assumptions that the set $\mathcal{F}$ is
\emph{full-dimensional} ($\mathcal{F}$ has positive volume in $\QQ^V$) and
\emph{bounded} ($\mathcal{F}$ is contained within a ball of finite radius),
the weak optimization problem for $\mathcal{F}$ is solvable in polynomial time
if, and only if, the corresponding weak separation problem is solvable in polynomial
time. The following is essentially the statement of Theorem 4.2.7 in~\cite{GLS88}.

\begin{theorem}\label{thm:GLSopt}
	Let $V$ be a non-empty set, and $\mathcal{F}\subseteq \QQ^V$ a 
	full-dimensional convex set that is located inside the ball $\mathcal{B}(0,R)$
	for some known value $R$. The weak optimization problem on $\mathcal{F}$
	is solvable in polynomial time if its weak separation problem is solvable in
	polynomial time. 
\end{theorem}

Note that in the special case where $\mathcal{F}$ is a 
rational polyhedron (for instance the feasible region of a linear program),
even the strong versions of the problems can be solved
in polynomial time, and the additional assumptions just introduced can be avoided.

The main tool in the reduction from optimization to separation is the so-called
\emph{ellipsoid method}(see~\cite{GLS88}), which is an algorithm that repeatedly calls
a blackbox solver for the separation problem, called a \emph{separation oracle},
in order to locate a feasible point. Finding an algorithm for the
optimization problem then reduces to finding a suitable separation 
oracle for its feasible region. In the case of semidefinite
programs, a simple (strong) separation oracle is given by Algorithm \ref{alg:sdpsep}.
Note that this algorithm assumes that we can compute the eigenvalues and eigenvectors
with infinite precision. Realistically we will have to work with finite
precision approximations that can be obtained in polynomial time. For our purposes,
in Section~\ref{sec:sep}, we will formulate a slightly modified algorithm that serves as
a weak separation oracle that can be expressed in fixed-point logic with counting (\FPC).

Anderson et al.~\cite{adh13:lics} have shown that the
ellipsoid method can be suitably expressed in fixed-point logic with counting, 
at least in the case of polyhedra.  In the present paper we extend their construction
 to semidefinite sets.  Together with the \FPC-definable
weak separation oracle, this yields our definability result for semidefinite 
programs. 

\begin{algorithm}
	\caption{Separation oracle for semidefinite programs}\label{alg:sdpsep}
	\begin{algorithmic}[1]
		\Require $\mathcal{A}=\{A_1,\ldots,A_m\in \QQ^{V\times V}\}$, $b\in\QQ^m$, $Y\in\QQ^{V\times V}$.
		\Ensure Solves separation problem on $\mathcal{F}_{\mathcal{A},b}$ and $Y$.
		\Function{Separation}{$\mathcal{A}$,$b$,$Y$}:
			\If{there is $A_i\in \mathcal{A}$ such that $\innerprod{A_i}{Y}>b_i$} \label{line:select1}
				\State \Return $\frac{1}{\norm{A_i}}A_i$
			\EndIf
			\State Compute Eigenvalues $\{\lambda_1,\ldots,\lambda_{|V|}\}$ of $Y$
			\If{there is $\lambda_i<0$}
				\State $v\gets $ Eigenvector corresponding to $\lambda_i$
				\State \Return $(-1)/\norm{vv^T}\cdot vv^T$
			\EndIf
			\State \Return \textsc{Accept}
		\EndFunction
	\end{algorithmic}
\end{algorithm}

\subsection{Lasserre Hierarchy}
One of the most common applications of semidefinite programming is to give 
approximation algorithms for hard combinatorial problems. For a large class of
problems, namely those that can be expressed as integer programs, a 
generic way to find approximations is to drop the integrality condition so that
a rational solution can be efficiently computed. The value
of the optimal rational solution serves as an upper bound to the optimal value
of the integer problem and can be used as an approximation. 
The concept of relaxation hierarchies extends this idea further. Instead of 
solving the basic relaxation, these hierarchies define a sequence of linear or 
semidefinite programs that provide increasingly finer approximations to the 
original integer solution. This is achieved by adding, at each level of the 
hierarchy, additional constraints to the basic relaxation that are preserved 
under the original integer 
program, but may cut away rational solutions not in the convex hull of integer 
solutions.

A prominent example of such a relaxation hierarchy is the 
\emph{Lasserre hierarchy} which for a given 0--1 linear program
defines a sequence of semidefinite programs.

\begin{definition}
	Let $V$, $M$ be sets and  $\mathcal{K}:=\{x\in\QQ^V\mid Ax\geq b\}$ a polytope  given by $A\in\QQ^{U\times V}, b\in\QQ^U$.
	
	For a vector $y\in\QQ^{\powerset(V)}$, and an integer $t$ with $1 \leq t \leq |V|$,
	we define the \emph{$t$-th moment matrix} of $y$, $M_t(y)$ as the 
	$\powerset_t(V)\times \powerset_t(V)$-matrix with entries
	\[M_t(y)_{I,J}:= y_{I\cup J}, \text{ for } |I|,|J|\leq t.\]
	
	Similarly, the \emph{$t$-th moment matrix of slacks} of $y, A, b$, and some
	$u\in U$ is given by
	\[S_t^u(y)_{I,J}:=\sum_{v\in V}A_{u,v}y_{I\cup J\cup \{v\}}-b_u y_{I\cup J},
	\text{ for } |I|,|J|\leq t.\]
	
	Finally, the \emph{$t$-th level of the Lasserre hierarchy} of $\mathcal{K}$,
	$\Las_t(\mathcal{K})$ is the positive semidefinite set defined by
	\[\Las_t(\mathcal{K}):=\{y\in\QQ^{\powerset_{2t+1}(V)}\mid y_{\emptyset}=1, M_t(y)\succeq 0, S_t^u(y)\succeq 0 \text{ for all } u\in U\}.\]
	We write $\Las_t^\pi(\mathcal{K}):=\{y_{\{v\}}, v\in V\mid y\in\Las_t(\mathcal{K})\}$
	for the projection of $\Las_t(\mathcal{K})$ onto the original variables.
\end{definition}

The general usage of the Lasserre hierarchy is as follows.
Assume we have a 0--1 program where the feasible region is defined as
$\mathcal{K}\cap \{0,1\}^{V}$. Now, instead of optimizing over the integer region,
we can define $\Las_t(\mathcal{K})$ for some level $t$, and solve the corresponding SDP. 
For a fixed constant $t$, this SDP has a polynomial number of new variables, and the
optimum can be obtained in polynomial time.  This optimum, when projected down onto
the original variables, serves as an approximation to the optimum in 
$\mathcal{K}\cap \{0,1\}^{V}$. 

The following basic properties of the Lasserre hierarchy establish that $\Las_t(\mathcal{K})$
is indeed a relaxation of $\mathcal{K}\cap \{0,1\}^{V}$. We write $\mathcal{K}^*$
for the polytope that is defined by the convex hull of the integer points in $\mathcal{K}$,
i.e.\ $\mathcal{K}^*:=\conv(\mathcal{K}\cap \{0,1\}^{V})$.

\begin{lemma}
	Let $\mathcal{K}=\{x\in\QQ^V\mid Ax\geq b\}$, and $y\in\Las_t(\mathcal{K})$
	for $t\in \{1,\ldots,|V|\}$. Then,
	\begin{enumerate}
	  \item $\mathcal{K}^*\subseteq \Las_t^\pi(\mathcal{K})$.
	  \item $\Las_0(\mathcal{K})\supseteq\Las_1(\mathcal{K})\supseteq \ldots\supseteq\Las_{|V|}(\mathcal{K})$.
	  \item $\Las_0^\pi(\mathcal{K})\subseteq \mathcal{K}$, and $\mathcal{K}^* = \Las_{|V|}^\pi(\mathcal{K})$.
	\end{enumerate}
\end{lemma}
\begin{proof}
	See for instance in \cite{R13}
\end{proof}

\begin{definition}
	Let $I=(c,\mathcal{K})$ be a 0--1 linear program with a objective vector $c\in\QQ^V$
	optimizing over a feasible region $\mathcal{K}\cap \{0,1\}^{V}$. 
	We say that $I$ is \emph{captured} at the $t^{\scriptsize\mbox{th}}$ level
	of the Lasserre hierarchy if $\Las_t^\pi(\mathcal{K})=\mathcal{K}^*$.
	We write $l(I)$ for the minimum $t$, such that $I$ is captured at 
	the $\tth$ level.
	
	For a class of 0--1 linear programs $C$, we say that $C$ is
	captured at the $\tth$ level, if every program in $C$ is captured at the 
	$\tth$ level of the Lasserre hierarchy.
	
	We denote by $L_C(n)$ the function that maps an integer $n$ to the 
	lowest level $t$ at which any 0-1 program in $C$ with size $n$ is captured.
	That is, $L_C(n):=\max_{I\in C;|V|\leq n}l(I)$. 
\end{definition}

We see that at a sufficiently high level, namely at most at level $t=|V|$, any 0--1
program is captured by the $\tth$ level Lasserre relaxation.
In those cases, the optimum of the Lasserre set yields not only an 
approximate optimum, but is the exact optimal value of the original 0--1 problem. 

Note that it could in principle occur that some instance $I=(c,\mathcal{K})$ 
is solved exactly by a Lasserre relaxation of level $t$ strictly less than $l(I)$. 
That is, $\mathcal{K}^*$ is strictly contained in $\Las_t^\pi(\mathcal{K})$, 
but they coincide in direction of the objective vector $c$. We argue that for
classes of 0--1 programs for which the objective $c$ can be chosen arbitrarily,
$L_C(n)$ still is a right notion for the number of Lasserre levels required to
solve instances of $C$ exactly. This is formalized in the lemma below.

\begin{lemma}\label{lem:las_capture}
	For any 0--1 linear program $I=(c,\mathcal{K})$, let $f(I)$ denote the 
	minimum $t$, such that 
	$\max_{x\in \mathcal{K}^*}\innerprod{c}{x}=\max_{x\in\Las_t^\pi(\mathcal{K})}\innerprod{c}{x}$.
	
	Furthermore, let $C$ be a class of 0--1 linear programs, such that, if 
	$I=(c,\mathcal{K})$ is an instance of $C$, then $J=(c',\mathcal{K})$
	is also an instance of $C$ for any choice of $c'\in\QQ^V$.
	
	Then, $L_C(n) = \max_{I\in C;|V|\leq n}f(I)$.
\end{lemma}
\begin{proof}
	Assume the claim is false. Then we can pick an $n$, such that 
	$L_C(n) > \max_{I\in C;|V|\leq n}f(I)$, and define $k:=\max_{I\in C;|V|\leq n}f(I)$.
	Now pick an instance $I=(c,\mathcal{K})$ that maximizes the right hand 
	side, i.e.\ $f(I)=k$. In particular, $l(I)>k$, and there is some direction $c'$ 
	where the maxima of the $k$th Lasserre relaxation and the convex hull of
	integer solutions do not match, i.e.\ 
	$\max_{x\in \mathcal{K}^*}\innerprod{c'}{x}<\max_{x\in\Las_k^\pi(\mathcal{K})}\innerprod{c'}{x}$.
	However by assumption the instance $J=(c',\mathcal{K})$ is also in $C$, 
	and $f(J)$ must be larger than $k$. This is a contradiction to the definition
	of $k$.
\end{proof}

Note that the class of $\VCSP$s satisfies the condition of the above lemma: If
an instance of $\VCSP(D,\Gamma)$ is represented by a 0--1 linear program 
$(c,\mathcal{K})$, then any other cost vector $c'$ can be achieved by changing
the weights on the constraints.

Our main result establishes a dichotomy for $\VCSP$s with respect to $L_C(n)$: 
For every $(D,\Gamma)$, $C=\VCSP(D,\Gamma)$, either $L_C(n)=0$ ($\VCSP(D,\Gamma$
is solved by the basic linear program relaxation); or $L_C(n)\in\Omega(n)$.

\subsection{Logic}\label{sec:logic}

We define the logical notions used throughout the paper. 

A relational \emph{vocabulary} $\tau$ is a finite sequence of relation and constant symbols 
$(R_1, \dots, R_k, c_1, \dots, c_l)$, where every relation symbol $R_i$ has a
fixed arity $a_i \in \NN$. A structure
$\struct A = (\univ A, R_1^{\struct A}, \dots, R_k^{\struct A}, c_1^{\struct A}, \dots, c_l^{\struct A})$
over the signature $\tau$ (or a \emph{$\tau$-structure}) consists of a non-empty
set $\univ A$, called the \emph{universe} of $\struct A$, together with
relations $R_i^{\struct A} \subseteq \univ A^{a_i}$ and constants
$c_j^{\struct A} \in \univ A$ for each $1 \leq i \leq k$ and $1 \leq j \leq l$.
Members of the set $\univ A$ are called the \emph{elements} of $\struct A$ and 
we define the \emph{size} of $\struct A$ to be the cardinality of its universe,
often written as $|\struct A|$.
	
\subsubsection{Fixed-point Logic with Counting}

Fixed-point logic with counting (\FPC) is an extension of inflationary 
fixed-point logic with the ability to express the cardinality of definable sets.   Here we give a bare-bones definition of the logic.  For more details, we refer the reader
to~\cite{EF99, Lib04}. 
The logic has two sorts of first-order variables: \emph{element variables}, 
which range over elements of the structure on which a formula is interpreted in 
the usual way, and \emph{number variables}, which range over some initial 
segment of the natural numbers. We usually write element variables with lower-case 
Latin letters $x, y, \dots$ and use lower-case Greek letters $\mu, \eta, \dots$ 
to denote number variables.

The atomic formulas of $\FPC[\tau]$ are all formulas of the form $\mu
= \eta$ or $\mu \le \eta$, where $\mu, \eta$ are number variables; $s
= t$ where $s,t$ are element variables or constant symbols from
$\tau$; and $R(t_1, \dots, t_m)$, where each $t_i$ is either an
element variable or a constant symbol and $R$ is a relation symbol
(i.e.\ either a symbol from $\tau$ or a relational variable) of arity
$m$.  Each relational variable of arity $m$ has an associated type
from $\{\mathrm{elem},\mathrm{num}\}^m$.  The set $\FPC[\tau]$ of 
\emph{$\FPC$ formulas} over $\tau$ is built up from the atomic formulas by 
applying an inflationary fixed-point operator $[\ifpop_{R,\tup x}\phi](\tup t) $;
forming \emph{counting terms} $\countingTerm{x} \phi$, where $\phi$ is a formula
and $x$ an element variable; forming formulas of the kind $s = t$ and $s \le t$ 
where $s,t$ are number variables or counting terms; as well as the standard 
first-order operations of negation, conjunction, disjunction, universal and 
existential quantification. Collectively, we refer to element variables and 
constant symbols as \emph{element terms}, and to number variables and counting 
terms as \emph{number terms}.

For the semantics, number terms take  values in $\{0,\ldots,n\}$,
where $n= \univ{A}$ and
element terms take values in $\dom(\struct A)$. The semantics of atomic formulas,
fixed-points and first-order operations are defined as usual (c.f.,
e.g., \cite{EF99} for details), with comparison of number terms
$\mu \le \eta$ interpreted by comparing the corresponding integers in
$\{0,\ldots,n\}$. Finally, consider a counting term of the form
$\countingTerm{x}\phi$, where $\phi$ is a formula and $x$ an element
variable. Here the intended semantics is that $\countingTerm{x}\phi$
denotes the number (i.e.\ the element of $\{0,\ldots,n\}$) of elements that
satisfy the formula $\phi$.

Throughout the paper, we make frequent use of the Immerman-Vardi theorem~\cite{EF99},
which establishes that fixed-point logic can
express all polynomial-time properties of finite ordered structures.  
It follows that in $\FPC$ we can express all polynomial-time relations on the 
number domain. 

We write $\Ck$ for the fragment of first-order logic with counting quantifiers 
consisting of those formulas that can be written using at most $k$ distinct 
variables.  It is easy to see  that any structure with $n$ elements can be 
described up to isomorphism by a formula using no more than $n$ variables. 
It follows that any collection of structures, each of which has no more than $n$ 
elements,  can also be characterized up to isomorphism by a formula with no more 
than $n$ variables. 

The minimum number of variables needed to define a class of structures in $\Ck$
turns out to be a useful measure of complexity. This motivates
the definition of the \emph{counting width} of a class.

\begin{definition}\label{def:counting-width}
For any class of structures $\mathcal{C}$, the \emph{counting width} of 
$\mathcal{C}$ is the function $\nu_\mathcal{C}:\NN\rightarrow\NN$ where  $\nu_\mathcal{C}(n)$ is the minimum value $k$ such that
	there is a formula $\phi$ in $\Ck$, for which any structure $\struct A$ with 
	$|\dom(\struct{A})|\leq n$, it holds $\struct{A}\models \phi\Leftrightarrow \struct{A}\in\mathcal{C}$.
\end{definition}

It is clear that $\nu_{\mathcal{C}} = \Omega(n)$ for any class $\mathcal{C}$.
It is known that if $\mathcal{C}$ is definable in \FPC, then $\nu_\mathcal{C}$ 
is bounded by a constant (see~\cite{Ott97}). 
The converse is not true in general as there are even undecidable classes $\mathcal{C}$ 
for which  $\nu_\mathcal{C}$ is bounded by a constant. 
However, the converse holds in special cases, such as for constraint satisfaction
problems. Here we have a dichotomy: every $\mathcal{C}=\CSP(D,\Gamma)$ is either definable in \FPC or 
has unbounded $\nu_\mathcal{C}$.  For an explanation see~\cite{DW15} where this 
result is extended to finite valued CSPs.  

In Section~\ref{sec:linbound}, we will show that the counting width of finite-valued
$\CSP$s is either bounded by a constant, or is $\Omega(n)$.  We use this for 
our main result to establish a similar dichotomy on  the number of levels 
of the Lasserre hierarchy needed to capture the 0--1 linear programs 
coding instances of $\VCSP$s.

\subsubsection{Interpretations}
	
We frequently consider ways of defining one structure within another in some 
logic $\logic L$, such as first-order logic or $\FPC$. 
Consider two signatures $\sigma$ and $\tau$ and a logic $\logic L$. 
An \emph{$m$-ary $\logic L$-interpretation of $\tau$ in $\sigma$} 
is a sequence of formulae of $\logic L$ in vocabulary $\sigma$ consisting of:
(i) a formula $\delta(x)$;
(ii) a formula $\varepsilon(x, y)$;
(iii) for each relation symbol $R \in \tau$ of arity $k$, a formula $\phi_R(x_1, \dots, x_k)$; and
(iv) for each constant symbol $c \in \tau$, a formula $\gamma_c(x)$,
where each $x$, $y$ or $x_i$ is an $m$-tuple of free variables. 
We call $m$ the \emph{width} of the interpretation. 
We say that an interpretation $\Theta$ associates a $\tau$-structure 
$\struct B$ to a $\sigma$-structure $\struct A$ if there is a 
surjective map $h$ from the $m$-tuples 
$\{ a \in \univ{A}^m \mid \struct A \models \delta[a] \}$ 
to $\struct B$ such that:

\begin{itemize}
	\item $h(a_1) = h(a_2)$ if, and only if, $\struct A \models \varepsilon[a_1, a_2]$;
	
	\item $R^\struct{B}(h(a_1), \dots, h(a_k))$ if, and only if, $\struct A \models \phi_R[a_1, \dots, a_k]$;
	
	\item $h(a) = c^\struct{B}$ if, and only if, $\struct A \models \gamma_c[a]$.
\end{itemize}

\noindent 
Note that an interpretation $\Theta$ associates a $\tau$-structure with 
$\struct A$ only if $\varepsilon$ defines an equivalence relation 
on $\univ{A}^m$ that is a congruence with respect to the relations defined by 
the formulae $\phi_R$ and $\gamma_c$. In such cases, however, $\struct B$ is 
uniquely defined up to isomorphism and we write $\Theta(\struct A) \defeq \struct B$. 
Throughout this paper, we will often use interpretations where $\varepsilon$ is simply defined
as the usual equality on $a_1$ and $a_2$. In these instances, we omit the explicit
definition of $\varepsilon$.

The notion of interpretations is used to define logical reductions. Let $C_1$ and $C_2$
be two classes of $\sigma$- and $\tau$-structures respectively. We say that $C_1$ \emph{$\logic L$-reduces}
to $C_2$ if there is an $\logic L$-interpretation $\Theta$ of $\tau$ in $\sigma$, such that
$\Theta(\struct{A})\in C_2$ if and only if $\struct{A}\in C_1$, and we write $C_1\leq_{\logic L}C_2$.

It is not difficult to show that formulas of \FPC compose with \FPC-reductions
in the sense that, given an interpretation $\Theta$ of $\tau$ in $\sigma$ and 
a $\tau$-formula $\phi$, we can define a
$\sigma$-formula $\phi'$ such that $\struct A \models \phi'$ if, and
only if, $\Theta(\struct A) \models \phi$. Note that if $\phi$ uses $k$ variables,
the composition $\phi'$ may contain up to $m\cdot k$ many variables, where $m$ is
the width of $\Theta$.
Likewise, interpretations themselves
compose. That is, given interpretations $\Theta$ of $\tau$ in $\sigma$, and
$\Sigma$ of $\sigma$ in $\rho$, we can obtain an interpretation $\Theta'$ of
$\tau$ in $\rho$ by composition: $\Theta'$ consists of the functions of $\Theta$
where the relation symbols of $\sigma$ are instead replaced by the corresponding
$\rho$-formulas in $\Sigma$. 

Finally, dealing with \FPC-reductions allows us to track counting width in the following way.

\begin{proposition}\label{prop:fpc_nu}
	Let $C_1$ and $C_2$ be two classes of structures, such that $C_1\leq_{\FPC}C_2$
	by some \FPC-reduction $\Theta$. Furthermore, let $\theta:\NN\rightarrow\NN$
	be defined as $\theta(n)=\max_{\struct{A}\in C_1; |\struct{A}|\leq n}|\Theta(\struct{A})|$.
	Then $\nu_{C_1}(n)\in O(\nu_{C_2}(\theta(n)))$.
\end{proposition}
\begin{proof}
	Given any structure $\struct A$ (in the vocabulary of $C_1$) of size $n$, the corresponding
	structure $\Theta(\struct A)$ has size at most $\theta(n)$. Let $k:=\nu_{C_2}(\theta(n))$,
	then there is a formula $\phi$ in $C^k$ for which it holds 
	$\Theta(\struct A)\models \phi\Leftrightarrow \Theta(\struct A)\in C_2$. By
	composing $\phi$ with $\Theta$, we obtain a formula $\phi'$ in $C^{mk}$
	that satisfies $\struct A\models \phi' \Leftrightarrow \struct A\in C_1$, where
	$m$ is the width of $\Theta$. This constant factor is accounted 
	for in the $O$-notation.
\end{proof}

\subsubsection{Representation}
In order to discuss definability of constraint satisfaction and optimization
problems, we need to fix a representation of instances of these problems as 
relational structures. Here, we describe the representation we use, adapted from~\cite{adh15:jacm}.

\textbf{Numbers and Vectors.} We represent an integer $z$ as a 
relational structure in the following way. Let $z=s\cdot x$, with $s\in\{-1,1\}$
being the sign of $z$, and $x\in\mathbb{N}$, and let $b\geq \lceil\log_2(x)\rceil$.
We represent $z$ as the structure $\struct{z}$ with universe $\{1,\ldots,b\}$
over the vocabulary $\tau_{\mathbb{Z}}=\{X,S,<\}$, where $<$ is interpreted 
the usual linear order on $\{1,\ldots,b\}$; $S^{\struct{z}}$ is a unary relation
where $S^{\struct{z}}=\emptyset$ indicates that $s=1$, and $s=-1$ otherwise;
and $X^{\struct{z}}$ is a unary relation that encodes the bit representation
of $x$, i.e.\ $X^{\struct{z}}=\{k\in\{1,\ldots,b\} \mid \BIT(x,k)=1\}$. In a similar
vein, we represent a rational number $q=s\cdot \frac{x}{d}$ by a structure
$\struct{q}$ over the domain $\tau_{\mathbb{Q}}=\{X,D,S,<\}$, where the
additional relation $D^{\struct{q}}$ encodes the binary representation of the
denominator $d$ in the same way as before.

In order to represent vectors and matrices over integers or rationals, we
have multi-sorted universes.
Let $T$ be a non-empty set, and let $v$ be a 
vector of integers indexed by $T$.  We represent $v$ as a structure $\struct{v}$
with a two-sorted universe with an index sort  $T$, and bit 
sort $\{1,\ldots,b\}$, where $b\geq\lceil\log_2(|m|)\rceil$,
$m=\max_{t\in T}v_t$, over the vocabulary $(X,D,S,<)$. Now, the
relation $S$ is of arity $2$, and $S^{\struct{v}}(t,\cdot)$ encodes the sign
of the integer $v_t$ for $t\in T$.  Similarly, $X$ is a binary relation interpreted
as $X^{\struct{v}}=\{(t,k)\in T\times \{1,\ldots,b\}\mid \BIT(v_t,k)=1\}$.
In order to represent 
matrices $M\in \mathbb{Z}^{T_1\times T_2}$, indexed by two sets $T_1, T_2$,
we have three-sorted universes with
two sorts of index sets, or simply a single index set that consists of pairs.
The generalization to rationals carries over from 
the numbers case. 
We write $\vocvec$ to denote the vocabulary for vectors over $\QQ$ and $\vocmat$ for the vocabulary for matrices over $\QQ$.

\textbf{Linear and Semidefinite Programs.} We represent linear or semidefinite
programs in their respective standard forms in the following way.
An instance of a linear program in standard form is 
given by a constraint matrix $A\in \mathbb{Q}^{M\times V}$,
and vectors $b\in\mathbb{Q}^M, c\in\mathbb{Q}^V$. 
Hence, we represent it as a structure over
the vocabulary $\tau_{\LP}=\tau_{vec}\dunion\tau_{mat}$. 

Likewise, a semidefinite program in conic standard form is specified by a
matrix-valued vector $\mathcal{A}\in\QQ^{M\times(V\times V)}$, an objective matrix
$C\in\QQ^{V\times V}$, and a vector $b\in\QQ^M$. This is again represented as
a structure over $\tau_{\SDP}=\tau_{vec}\dunion\tau_{mat}$. Sometimes it is more 
convenient to consider an \SDP\ in inequality standard form, which is specified
by a matrix $Z\in\QQ^{M\times M}$, a matrix-valued vector $\mathcal{Y}\in\QQ^{V\times(M\times M)}$ and an
objective vector $c\in\QQ^V$. Note that the vocabulary for
both representations are the same, and that the conversion between
the two standard forms can be expressed as an \FPC\ interpretation, as it only involves simple
substitution and rearranging of variables.

We can now state the definability result from~\cite{adh13:lics}, to the 
effect that there is an $\FPC$ interpretation that can define solutions to 
linear programs. We will show a generalization of the result to semidefinite 
programs in Theorem~\ref{thm:sdpdefinable}.
\begin{theorem}[Theorem 11, \cite{adh13:lics}]\label{thm:lpdefinable}
	There is an $\FPC$-interpretation $\Phi$ of $\tau_\QQ \dunion \vocvec$ in $\tau_{LP}$ that
	does the following:
	
	Let instances of a linear program be given by $(A,b,c)$ with
	$A\in \mathbb{Q}^{M\times V}$, $b\in\mathbb{Q}^M$, and $c\in\mathbb{Q}^V$.
	Its feasible region is denoted by $\mathcal{F}_{A,b}$.
	Let $\struct I$ be the relational representation of this LP.
	
	Then, $\Phi(\struct I)$ defines a relational representation of 
	$(f,v)$, with $f\in\QQ$, $v\in\QQ^V$, such that 
	\begin{itemize}
	  \item $f=1$ if, and only if, $\max_{x\in \mathcal{F}_{A,b}} c^Tx$ is unbounded;
	  \item $v\notin \mathcal{F}_{A,b}$ if, and only if, there is no feasible solution;
	  \item and $f=0, v=\mathrm{argmax}_{x\in \mathcal{F}_{A,b}} c^Tx$ otherwise.
	\end{itemize}
\end{theorem}

\textbf{CSPs.} 
For a fixed domain $D$, and a constraint language $\Gamma$, we can represent
an instance of $\CSP(D,\Gamma)$ in a natural way. Namely, the vocabulary $\tau_{\CSP(\Gamma)}$
consists of all relations in $\Gamma$. An instance $I=(V,C)$ is then represented
as the $\tau_\Gamma$-structure $\struct I=(V, (R^{\struct I})_{R\in\Gamma})$, 
where the universe is set to the set of variables $V$, and 
$s\in R^{\struct I}$ if there is a constraint $c=(s,R)$ in the constraint 
set $C$.

For the finite-valued variant, we define the vocabulary $\tau_{\VCSP(\Gamma)}$
as $\tau_{\VCSP(\Gamma)}=\{(R_f)_{f\in\Gamma}, W, <\}$. An instance $I=(V,C)$
is then represented as a structure $\struct I$ with a three-sorted universe:
A sort for variables $V$; a sort of constraints $C$; and a bit sort $\{1,\ldots,b\}$
for some sufficiently large $b$. The relation $R^{\struct I}_f\subseteq V^{\ar(f)}\times C$ then contains 
a tuple $(s,c)$ if $C$ contains a constraint of the form $(s,f,w)$. Similarly,
the relation $W^\struct I\subseteq C\times {1,\ldots,b}$ encodes the weight of 
each constraint $c=(s,f,w)$ in the relational representation of integers, i.e.\ 
$W^I(c,\cdot)=\{k\in\{1,\ldots,b\}\mid \BIT(w,k)=1\}$. Finally, $<$ is again just
interpreted as the usual natural order on $\{1,\ldots,b\}$.

%% file: result.tex
Our main result establishes a dichotomy for $\VCSP$ problems: Either 
$\VCSP(D,\Gamma)$ is
tractable, and every instance is captured by its basic linear programming relaxation;
or there are instances that are only captured after $\Omega(n)$ 
levels of the Lasserre hierarchy, where $n$ is the size of the instance. 
As a special case, 
we obtain the same dichotomy for the class of $\MAXCSP$ problems.

In the following, recall that we write $L_C(n)$ to denote the minimum number $t$,
such that the Lasserre relaxation at level $t$ suffices to capture all instances
of $C$ of size at most $n$, and we use $\nu_C$ to denote the counting width of a
class $C$. For the sake of legibility, we use
$L_\Gamma$ as a shorthand for $L_{\VCSP(D,\Gamma)}$ and $\nu_{\Gamma}$ as 
shorthand for $\nu_{\VCSP(D,\Gamma)}$.

\begin{theorem}\label{thm:main}
	For any $\VCSP(D,\Gamma)$ either 
	every instance $I$ is solved by $\BLP(I)$; or 
	$L_\Gamma(n) \in \Omega(n)$.
\end{theorem}

\begin{corollary}
	For any $\MAXCSP(D,\Gamma)$ either 
	every instance $I$ is solved by $\BLP(I)$; or 
	$L_{\MAXCSP(D,\Gamma)}(n) \in \Omega(n)$.
\end{corollary}

The key technical lemma here is a bound that relates the level of the 
Lasserre hierarchy required to capture all instances
of a $\VCSP$ to the counting width of its class of decision problems.

\begin{lemma}\label{lem:lasserre_lb}
	For any	$\VCSP(D,\Gamma)$, it holds $L_\Gamma\in \Omega(\nu_\Gamma)$.
\end{lemma}

In addition, we prove a counting width dichotomy for $\VCSP(D,\Gamma)$.
This is achieved by connecting the results of ~\cite{DW15} and ~\cite{Atserias20091666}
to show a linear lower bound of $\nu_\Gamma$ for the hard cases of $\VCSP(D,\Gamma)$.

\begin{lemma}\label{lem:linear_lb}
	If there are instances $I$ of $\VCSP(D,\Gamma)$ that are not solved by 
	$\BLP(I)$, then $\nu_\Gamma(n)\in \Omega(n)$.
\end{lemma}

Given the above two lemmas, we obtain as a direct consequence Theorem~\ref{thm:main}.
Hence we devote the remaining sections to proving Lemmas~\ref{lem:lasserre_lb}
and \ref{lem:linear_lb}.

In Section~\ref{sec:linbound} we provide a proof of Lemma~\ref{lem:linear_lb}. 
The main observation is that the relevant reductions described in~\cite{DW15} 
are essentially linear in size, and that we can reduce solving linear systems of
equations over the two-element field to the hard cases of $\VCSP(D,\Gamma)$. In
turn,~\cite{Atserias20091666} shows that these systems of equations have linear
counting width.

The general idea to prove Lemma~\ref{lem:lasserre_lb} follows two main steps. 
In the first, we establish using Theorem~\ref{thm:sdpdefinable} that the optimum
value of a Lasserre SDP can be defined within \FPC, given its explicit vector-matrix
representation. This part will be proved in Section~\ref{sec:sdp}.

In Section~\ref{sec:lasserre}, we then show that from any explicitly given 0--1 LP,
we can define its explicit $t$-th level Lasserre relaxation by an \FPC-interpretation using
only $O(t)$ many variables. With the result of Section~\ref{sec:sdp} this means
that there is a \FPC-formula in $O(t)$ variables that defines
the solution to the given 0--1 program. On the other hand, if we know that the
solution to some 0--1 program can not be defined using fewer than $\nu_{\Gamma}(n)$
many variables, then this implies a lower bound for the value of $t$ of also $\Omega(\nu_{\Gamma}(n))$.
This then concludes the proof of Lemma~\ref{lem:lasserre_lb}.

%% file: linbound.tex
In this section we aim to provide a proof for Lemma~\ref{lem:linear_lb}. 
The main pieces of the argument are known results from the literature, and we
simply lay out how they together imply the claim.

We aim to show a linear lower bound for the counting width of those $\VCSP$s
that are not solved by the BLP relaxation. This aligns
with the dichotomy result of Thapper and \v{Z}ivn\'y (Theorem~\ref{thm:vcsp}). 
That is, if $\VCSP(D,\Gamma)$
is not solved by the BLP relaxation, we know that $\MAXCUT$ reduces to it. Our 
strategy is to show that (i) $\MAXCUT$ has linear counting width; and (ii) there
is a linear size $\FPC$-reduction from $\MAXCUT$ to $\VCSP(D,\Gamma)$, if it is not
solved by its BLP relaxation. By Proposition~\ref{prop:fpc_nu} this suffices to 
prove our claim.

For (i), we consider the problem $3\LIN$:
An instance of $3\LIN$ consists of a set of variables
$V$, and two sets of equations, $E_0$ and $E_1$. Each equation in $E_0$ has the
form $a\oplus b\oplus c=0$, where $\oplus$ denotes addition modulo 2, and 
$a,b,c\in V$. Similarly, each equation in $E_1$ has the form $a\oplus b\oplus c=1$.
The problem is then to determine whether there is an assignment $h:V\rightarrow\{0,1\}$
such that all equations are satisfied.

\begin{lemma}
	$\nu_{3\LIN}(n)\in\Omega(n)$.
\end{lemma}
\begin{proof}
	In~\cite{Atserias20091666} Atserias et al.\ show a lower bound of for the 
	counting width of the problem $3\LIN$ that is proportional to the tree-width
	of the instance. More precisely, they show a construction that transforms
	any given graph $G=(V,E)$ with tree-width $t$ into a pair of $3\LIN$ instances
	$(I,I')$, each having $O(|V|)$ variables, such that $I$ is satisfiable,
	but $I'$ is not, and no $\Ck$ formula of at most $t$ variables distinguishes
	between them.	
	
	The claim then follows by picking a class of graphs that have linear tree-width.
	Such graphs exist, for instance in the class of 3-regular expander 
	graphs~\cite{Ajtai1994}. (A similar argument of picking linear tree-width 
	graphs was already present in~\cite{CFI}).
\end{proof}

As a direct consequence, we obtain that $3\SAT$ also has linear counting width.

\begin{lemma}
	$\nu_{3\SAT}(n)\in\Omega(n)$.
\end{lemma}
\begin{proof}
	Given an instance $(V,E_0,E_1)$ of $3\LIN$, we replace each equation 
	$a\oplus b\oplus c=0$ by the four clauses containing $a,b,c$ that have an 
	even number of negated literals, i.e.\ 
	 $(a\vee b\vee c)$, $(\neg a\vee \neg b\vee c)$,
	$(a\vee\neg b\vee \neg c)$, and $(\neg a\vee b\vee\neg c)$. Similarly,
	each equation $a\oplus b\oplus c=1$ is replaced by the four clauses of $a,b,c$
	that have an odd number of negated literals. This results in a $3\SAT$ instance
	that is satisfiable if and only if the original $3\LIN$ instance was satisfiable.
	Clearly, this is a linear size reduction that can be implemented in \FPC.
\end{proof}

We continue the reduction to $\MAXCUT$.

\begin{lemma}
	$\nu_{\MAXCUT}(n)\in\Omega(n)$.
\end{lemma}
\begin{proof}
	In~\cite{DW15}, we find an explicit construction of a \FPC-reduction from 
	$3\SAT$ to $\MAXCUT$. This reduction is also linear size.
\end{proof}

Finally, the reduction for (ii) has already been explicitly constructed 
in~\cite{DW15}. It is not difficult to confirm that these reductions are
in fact linear in size. This chain of reductions then concludes the proof of 
Lemma~\ref{lem:linear_lb}.

%% file: sdp.tex
We now turn to our definability result for semidefinite programs which 
states that the weak optimization problem for explicitly given SDPs is
expressible in \FPC. Our result relies heavily on previous work by 
Anderson et al.~\cite{adh13:lics,adh15:jacm} for the case of linear programming. In fact, 
their proof method allows a simple 
adaptation:	The central piece there is a formulation of the ellipsoid method for
polyhedra in \FPC. That is, they show that the
reduction from the optimization problem to the separation problem for polyhedra 
can be accomplished in \FPC. They then show that the separation problem for
explicit LPs is also definable in \FPC. The same approach can be taken now for
the case of SDPs, where we aim to solve the weak formulations of the optimization and
separation problems. Formally, we prove the following.

\begin{theorem}\label{thm:sdpdefinable}
	There is an $\FPC$-interpretation $\Phi$ of $\vocmat$ in $\tau_{\SDP}\dunion \tau_\QQ$ that
	does the following:
	
	Let instances of a SDP be given by $(A,b,C)$ and an error parameter $\delta$,
	with
	$A\in \mathbb{Q}^{M\times (V\times V)}$, $b\in\mathbb{Q}^M$, 
	$C\in\mathbb{Q}^{V\times V}$, and $\delta>0$.
	Its feasible region is denoted by $\mathcal{F}_{A,b}$.
	Let $\struct I$ be the relational representation of this SDP.
	
	Then, $\Phi(\struct I)$ defines a relational representation of 
	$X\in\QQ^{V\times V}$, such that 
	\begin{itemize}
	  \item if $\mathcal{F}_{A,b}$ is empty or unbounded, there is no specification on $X$;
	  \item otherwise $X$ is a $\delta$-close and $\delta$-maximal solution.
	\end{itemize}
\end{theorem}

\subsection{Separation Oracle}\label{sec:sep}
Similar to the work in \cite{adh13:lics}, our proof strategy is to use of 
an \FPC-formulation of
the ellipsoid method to reduce the optimization problem to the separation problem.
Therefore in order to prove Theorem \ref{thm:sdpdefinable} we show first
that we can express a separation oracle for SDPs in \FPC. 

\begin{lemma}\label{lem:sepdefinable}
	There is an $\FPC$-interpretation $\Phi$ of $\vocmat$ in 
	$\tau_{\SDP}\dunion \tau_\QQ$ that
	does the following:
	
	An instance of the separation problem is given by $(A,b,Y)$, 
	and an error parameter $\delta$,
	with
	$A\in \mathbb{Q}^{M\times (V\times V)}$, $b\in\mathbb{Q}^M$, 
	$Y\in\mathbb{Q}^{V\times V}$, and $\delta>0$.
	Let $\struct I$ be the relational representation of this instance.
	
	Then, $\Phi(\struct I)$ defines a relational representation of 
	$S\in\QQ^{V\times V}$, such that 
	\begin{itemize}
	  \item if $\mathcal{F}_{A,b}$ is empty or bounded, then there is no specification on $S$;
	  \item otherwise if $S=0$, then $Y$ is $\delta$-close to $\mathcal{F}_{A,b}$;
	  \item otherwise $\innerprod{S}{Y}+\delta>\max\{\innerprod{S}{X}\mid X\in \mathcal{F}_{\mathcal{A},b}\}$.
	\end{itemize}
\end{lemma}

Algorithm \ref{alg:sdpsep} describes a simple algorithm for
the separation problem for the feasible region of SDPs. Its correctness follows 
from the fact that an infeasible point $Y$ has to violate some inequality
$\innerprod{A_i}{Y}\leq b_i$ or the constraint $Y\succeq 0$. In the former case
we can simply chose the separation normal as $A_i$, while in the latter case
we can choose the normal to be $(-1)vv^T$ where $v$ is an eigenvector corresponding
to a negative eigenvalue $\lambda$ of $Y$, since 
$\innerprod{vv^T}{Y}=\lambda\cdot \norm{v}^2$. In order to implement this
algorithm in $\FPC$ however, we have to make two key modifications: (1) As
we want our output to be a rational vector, we have to work with a finite precision
in the calculations, and (2) all the steps must be definable in $\FPC$. Together
this leads to Algorithm \ref{alg:sdpweaksep} that solves the weak separation problem
for SDPs, and is possible to translate to $\FPC$ using known techniques. 

\begin{algorithm}
	\caption{Weak separation oracle for semidefinite programs}\label{alg:sdpweaksep}
	\begin{algorithmic}[1]
		\Require $\mathcal{A}=\{A_1,\ldots,A_m\in \QQ^{V\times V}\}$, $b\in\QQ^m$, $Y\in\QQ^{V\times V}$, $\delta\in\QQ$ such that $\delta > 0$.
		\Ensure Solves weak separation problem on $\mathcal{F}_{\mathcal{A},b}$, $Y$, and $\delta$.
		\Function{Separation}{$\mathcal{A}$,$b$,$Y$,$\delta$}:
			\State $\mathcal{V} \gets \{A_i\in\mathcal{A} \mid \innerprod{A_i}{Y}>b_i\}$
			\If{$\mathcal{V}$ is non-empty}
				\State $v\gets \sum_{A_i\in\mathcal{V}} A_i$ \label{line:select}
				\State \Return $\frac{v}{\norm{v}}$
			\EndIf
			\State Approximate eigenvalues $\{\tilde\lambda_1,\ldots,\tilde\lambda_{|V|}\}$ of $Y$ up to precision $\frac{\delta}{4}$ \label{line:eigenvalues}
			\If{there is $\tilde{\lambda}$ with $\tilde\lambda_i<\frac{\delta}{2}$}
				\State $v\gets $ Vector satisfying $\norm{(Y-\tilde{\lambda_i}I)v}<\frac{\delta}{2}$ and $\norm{v}=1$ \label{line:eigenvector}
				\State \Return $(-1)/\norm{vv^T} \cdot vv^T$
			\EndIf
			\State \Return \textsc{Accept}
		\EndFunction
	\end{algorithmic}
\end{algorithm}

The translation of Algorithm \ref{alg:sdpweaksep} into an \FPC-interpretation 
uses some known tools from descriptive complexity. First, we note that the 
basic vector and matrix operations, such as addition, multiplication, norm and 
even computing the characteristic polynomial can all be defined in \FPC 
\cite{HolmPhD}. A key modification is in Line \ref{line:select}:
In the original algorithm, we had to choose a violated
constraint from an unordered set of constraints, which is in general not possible
to express in \FPC. However, we can employ the same technique as in \cite{adh13:lics}:
the explicit choice of a constraint can be avoided by summing all violated constraints,
since by linearity, the sum of violated constraints is again a violated constraint,
which in turn is expressible in \FPC.

In Line \ref{line:eigenvalues}, we compute the eigenvalues of the input
matrix $Y$ up to a given precision $\delta/4$. This is possible in \FPC since
it is powerful enough to define the coefficients of the characteristic polynomial
of definable matrices (see~\cite{HolmPhD}).

\begin{proposition}
	There is an $\FPC$ interpretation of $\tau_{\QQ}$ in $\vocmat\dunion\tau_{\QQ}$
	that for a given a matrix $A\in\QQ^{V\times V}$ and a value $\delta\in\QQ$
	(in their relational representation) defines the value of the smallest 
	eigenvalue of $A$ up to a precision of $\delta$.
\end{proposition}
\begin{proof}
	Holm~\cite{HolmPhD} establishes that there is an interpretation in $\FPC$
	by which we can obtain from $A$ the coefficients $\alpha_1,\ldots,\alpha_n$
	of the characteristic polynomial 
	$p(x)=\mathrm{det}(xI-A)=x^n-\alpha_1x^{n-1}+\ldots+(-1)^n\alpha_n$. Since
	the coefficients have a linear order, by the Immerman-Vardi theorem, any
	polynomial time computable property can be defined in $\FPC$, such as
	computing the smallest eigenvalue up to a precision $\delta$.
\end{proof}


Furthermore, the exact calculation of the eigenvector corresponding to a negative eigenvalue
has been replaced by a linear optimization step in Line \ref{line:eigenvector}.
In general, the eigenvectors corresponding to some eigenvalue $\lambda$ are not
uniquely defined. Not only can we scale eigenvectors by an arbitrary amount, 
in the case of an eigenvalue of higher multiplicity we have to choose a 
representative from a whole multidimensional eigenspace. To avoid this choice,
we reformulate the problem as a linear program and rely on Theorem \ref{thm:lpdefinable}
to express this step in \FPC. (While this LP-step is used as a blackbox here,
in Section \ref{sec:opt} we give some exposition on the techniques in
\cite{adh13:lics} that are used to break the symmetry between choices.)

The correctness of the algorithm follows from some basic calculations. Assume
the algorithm accepts an input $(\mathcal{A},b,Y,\delta)$. Then $Y$ violated
none of the inequalities $\innerprod{A_i}{Y}\leq b_i$, and all eigenvalues
of $Y$ are non-negative, and we can conclude that $Y\in\mathcal{F}_{\mathcal{A},b}$.
Otherwise either some inequality $\innerprod{A_i}{Y}\leq b_i$ is violated,
in which case the algorithm produces a correct separating hyperplane, or some
approximated eigenvalue $\tilde\lambda$ is smaller than $\delta/2$. 
In the latter case, the linear optimization step looks
for a vector $v$ such that $\norm{v}=1$ and $\norm{(Y-\tilde\lambda)v}<\delta/2$.
Note that such a vector always exists. Let $\lambda$ be the actual eigenvalue, with
$\tilde{\lambda}=\lambda+\epsilon$ for some error $\epsilon$ with  $|\epsilon|\leq \delta/4$.
We have
$(Y-\tilde\lambda I)v = (Y-\lambda I-\epsilon I)v$, and by setting
$v$ to some eigenvector corresponding to $\lambda$ with $\norm{v}=1$,
we get $\norm{(Y-\tilde\lambda)v} = \norm{(Y-\lambda I - \epsilon I)v} = |\epsilon| < \delta/2$. 
Finally, given such a vector $v$, the normal of a weakly separating plane is given by
$S:=-1/\norm{vv^T}\cdot vv^T$: Let $\epsilon := Yv - \tilde\lambda v$.
Then $\innerprod{S}{Y}=-1/\norm{vv^T}\cdot(v^T Y v)=-1/\norm{vv^T}\cdot(v^T(\tilde\lambda v +\epsilon))=-1/\norm{vv^T}(\tilde\lambda+v^T\epsilon)<\delta$.

This shows that we can define a weak separation oracle for SDPs in \FPC. For the 
next step, we show that the reduction from weak optimization to separation,
i.e.\ the ellipsoid method, can be defined in \FPC as well.

\subsection{Reducing Optimization to Separation}\label{sec:opt}

In this section we construct a \FPC-reduction from the weak
optimization problem to the weak separation problem for SDPs.

\begin{lemma}\label{lem:optdefinable}
	If there is a \FPC-interpretation expressing the weak separation problem for
	the feasible region of a given SDP, then there is a
	\FPC-interpretation which expresses the
	weak semidefinite optimization problem.
\end{lemma}

A version of this lemma was already proved in \cite{adh13:lics} but was stated
in terms of the optimization and separation problems for polytopes. 
Their algorithm however generalizes nicely to our setting. Here, we give a 
brief overview of the main proof ideas again, and point out the changes we made
to accommodate the case of SDPs. For a detailed description of their algorithm, 
we refer to \cite{adh13:lics,adh15:jacm}.

The main idea behind the construction is to repeatedly apply
the separation oracle to define a linear order on the set of variables, and once a
sufficient order is obtained, to apply the Immerman-Vardi theorem to define the
ellipsoid method. This is achieved by defining a series of increasingly fine 
equivalence relations on the variable set $V$, specified by so-called 
\emph{foldings} that we formalize below. Intuitively, these partitions are
obtained as follows.
In the beginning, every element of $V$ resides in the same equivalence class. 
However, there may be some inputs on which the separation oracle returns a 
vector $d$ with different values $d_u$ and $d_v$ for $u,v\in V$, which
distinguishes the two elements $u$ and $v$. In subsequent iterations,
$u$ and $v$ are put in different equivalence classes, and this process is
repeated until we obtain a sufficiently refined partition of $V$.

There are a couple of key modifications
to be made to the algorithm from \cite{adh13:lics}. Namely, we show
(1) that the folding operation preserves
the positive semidefiniteness of sets, and (2) how to cope with the additional
parameters introduced by the weak versions of the separation and optimization
problems.

We start by defining the notion of \emph{folding}.
\begin{definition}
	Let $V$ be a non-empty set. For $k\leq |V|$, we call a surjective mapping
	$\sigma:V\rightarrow [k]$ an \emph{index map}. Furthermore, for each $i\in [k]$
	we define $V_i:=\{v\in V\mid \sigma(v)=i\}$.
	
	For a vector $x\in\QQ^V$, the \emph{almost-folded} vector $[x]^{\tilde{\sigma}}\in\QQ^k$
	is given by
	\[([x]^{\tilde{\sigma}})_i:=\sum_{v\in V_i}x_v, \text{ for } i\in [k].\]
	
	Its \emph{folded} vector $[x]^\sigma\in\QQ^k$ is given by
	\[([x]^\sigma)_i := [x]^{\tilde\sigma}_i/|V_i|, \text{ for } i\in [k].\]
	
	For a vector $\hat x\in\QQ^k$, its \emph{unfolded} vector $[\hat{x}]^{-\sigma}\in\QQ^V$
	is given by
	\[([\hat{x}]^{-\sigma})_v := \hat{x}_i, \text{ with } v\in V_i, \text{ for all } v\in V.\]
\end{definition}

For a given index map $\sigma$ and a vector $x\in\QQ^V$, we say $x$ \emph{agrees}
with $\sigma$ when for all $u,v\in V$ $\sigma(u)=\sigma(v)$ implies $x_u=x_v$. 
The notion also extends in a natural way to sets $S\subseteq\QQ^V$, simply by
defining the folded set $[S]^\sigma:=\{[s]^\sigma \mid s\in S\}$. This can be
seen as a projection of $S$ into the (ordered) $k$-dimensional space $\QQ^k$.
When talking about matrices, that is, when the variable set $V$ consists of pairs
from some product set $V'\times V'$, we implicitly also require that $\sigma$ is 
\emph{consistent}. Namely we require that
an index map $\sigma:V'\times V'\rightarrow [k]\times [k]$ is defined by an 
underlying index map $\tau:V'\rightarrow [k]$, with $\sigma(u,v)=(\tau(u),\tau(v))$.

There are some useful properties of the folding operation that allow us to infer
some information about the geometry of a folded set from its original.

\begin{proposition}\label{prop:innerprod}
	Let $\sigma:V\rightarrow [k]$ be an index map, $x, c$ vectors in $\QQ^V$,
	where $c$ agrees with $\sigma$. Then,
	\[\innerprod{c}{[[x]^\sigma]^{-\sigma}}=\innerprod{c}{x}=\innerprod{[c]^{\tilde\sigma}}{[x]^\sigma}.\]
\end{proposition}
\begin{proof}
	See Proposition 12 from \cite{adh13:lics}.
\end{proof}

\begin{proposition}\label{prop:fold}
	Let $\mathcal{P}\subseteq\QQ^V$ be a polytope in $\QQ^V$ and let 
	$\sigma:V\rightarrow [k]$ be an index map. Then the folded set $[\mathcal{P}]^\sigma$
	is a polytope in $\QQ^k$.
\end{proposition}
\begin{proof}
	See Proposition 13 from \cite{adh13:lics}.
\end{proof}

\begin{proposition}\label{prop:fold-sdp}
	Let $X\in\QQ^{V\times V}$ be a positive semidefinite matrix, and let 
	$\sigma:V\times V\rightarrow [k]\times [k]$. Then $[X]^\sigma$
	is also a positive semidefinite matrix in $\QQ^{k\times k}$.
\end{proposition}
\begin{proof}
	Since we assume the index map $\sigma$ to be consistent, we have a map
	$\tau:V\rightarrow [k]$ with $\sigma(u,v)=(\tau(u),\tau(v))$. Furthermore, 
	since $X$ is positive semidefinite, there exists vectors $g_v\in\QQ^l$ for all
	$v\in V$ and some $l\geq 1$ such that $X_{u,v}=\innerprod{g_u}{g_v}$ for all $u,v\in V$ 
	(i.e.\ the Gram representation of $X$). Let us now define vectors $g^\tau_1,\ldots,g^\tau_k$
	by \[g^\tau_i = \frac{1}{|V_i|}\sum_{v\in V_i}g_v,\]
	where $V_i:=\{v\in V\mid \tau(v)=i\}$. The vectors obtained in this way now 
	form a Gram representation of the folded matrix $[X]^\sigma$, since
	\[[X]^\sigma_{i,j}=\frac{1}{|V_{i,j}|}\sum_{(u,v)\in V_{i,j}} X_{u,v} 
		= \frac{1}{|V_i||V_j|}\sum_{u\in V_i}\sum_{v\in V_j} \innerprod{g_u}{g_v}
		= \innerprod{g^\tau_i}{g^\tau_j}.\]
	As the existence of a Gram representation implies positive semidefiniteness,
	this proves our claim.
\end{proof}

Since the feasible region of an SDP is the intersection of a polytope with
the positive semidefinite cone, Propositions~\ref{prop:fold} and~\ref{prop:fold-sdp} show that the result of 
folding the feasible region of an SDP is again the feasible region of an SDP. 
Next we show that a weak separation oracle of the original set either serves as an oracle
for the folded set, or produces some vector that does not agree with the index
map of the folding.

\begin{proposition}\label{prop:d-close}
	Let $\mathcal{F}\subseteq\QQ^V$ be a convex set, and let $\sigma:V\rightarrow [k]$
	be an index map. Given a vector $x\in\QQ^V$ that is $\delta$-close to $\mathcal{F}$
	for some $\delta\geq 0$, the folded vector $[x]^\sigma\in\QQ^k$ is also
	$\delta$-close to the folded set $[\mathcal{F}]^\sigma$.
\end{proposition}
\begin{proof}
	Let $x=f+d$, where $f\in\mathcal{F}$ is some point in the set $\mathcal{F}$,
	and $d$ the difference vector with $\norm{d}\leq \delta$. By the definition
	of folding, we then have $[x]^\sigma=[f]^\sigma+[d]^\sigma$, where $[f]^\sigma$
	is now a point in the folded set $[\mathcal{F}]^\sigma$. We can bound
	the norm of $[d]^\sigma$ by 
	\[\norm{[d]^\sigma} = \sqrt{\sum_{i\in [k]}\left(\frac{1}{|V_i|}\sum_{v\in V_i}d_v\right)^2}
		\leq \sqrt{\sum_{i\in [k]}(\max_{v\in V_i}d_v)^2}
		\leq \sqrt{\sum_{v\in V} d_v^2} = \norm{d}.\]
	Since $\norm{d}\leq \delta$, we have $\norm{[d]^\sigma}\leq \delta$.
\end{proof}

\begin{proposition}\label{prop:d-sep}
	Let $\mathcal{F}\subseteq\QQ^V$ be a convex set, and let $\sigma:V\rightarrow [k]$
	be an index map. Given vectors $s,y\in\QQ^V$ where $s$ agrees with $\sigma$, and
	$\innerprod{s}{y}+\delta >\max\{\innerprod{s}{x}\mid x\in\mathcal{F}\}$,
	it holds that 
	$\innerprod{[s]^\sigma}{[y]^\sigma}+\delta >\max\{\innerprod{[s]^\sigma}{x}\mid x\in\mathcal{[F]^\sigma}\}$.
\end{proposition}
\begin{proof}
	Let $x\in\mathcal{F}$ be a point in $\mathcal{F}$ such that $\innerprod{s}{x}$
	is maximal. It follows from Proposition \ref{prop:innerprod} that $[x]^\sigma$
	is also a maximal point in $[\mathcal{F}]^\sigma$ with respect to
	$\innerprod{[s]^\sigma}{[x]^\sigma}$.
	We then have
	\[\innerprod{[s]^\sigma}{[x]^\sigma}-\innerprod{[s]^\sigma}{[y]^\sigma}
		= \innerprod{[s]^\sigma}{[x-y]^\sigma}
		\leq \innerprod{s}{x-y} = \innerprod{s}{x} - \innerprod{s}{y} < \delta.\]
	For the first equality we use the fact that $[x-y]^\sigma=[x]^\sigma-[y]^\sigma$,
	and for the first inequality we use Proposition \ref{prop:innerprod} to
	get $\innerprod{[s]^\sigma}{[x]^\sigma}\leq \innerprod{[s]^{\tilde\sigma}}{[x]}=\innerprod{s}{x}$.
\end{proof}

Assume we are given a convex set $\mathcal{F}\subseteq\QQ^V$ by means of a 
corresponding weak separation oracle, and some index map 
$\sigma:V\rightarrow [k]$. Proposition \ref{prop:d-close}
ensures that whenever the oracle accepts some input $(y,\delta)$, then the folded
vector $[y]^\sigma$ is also $\delta$-close to the folded set $[\mathcal{F}]^\sigma$.
Likewise, by Proposition \ref{prop:d-sep} we know that whenever the oracle for $\mathcal{F}$
outputs a separation normal $s$ that agrees with $\sigma$, then the folded vector
$[s]^\sigma$ is also a $\delta$-weak separation normal that separates $[y]^\sigma$
from $[\mathcal{F}]^\sigma$. 

This leads us to a simple algorithm for the weak
separation problem for $[\mathcal{F}]^\sigma$: On some input $\hat{y}\in\QQ^k$
and $\delta\geq 0$, we simply give the input $[\hat{y}]^{-\sigma}$ and $\delta$
to the oracle for $\mathcal{F}$. If the oracle accepts, or returns a separation
normal that agrees with $\sigma$, we are done. In the other cases, the oracle
returns a separation normal that does not agree with our current index map
$\sigma$, in which case we can use that output to further refine the 
underlying equivalence relation. After at most $|V|$ such refinements we obtain
a correct weak separation oracle for $[\mathcal{F}]^\sigma$.

For the overall algorithm, we can now follow exactly the procedure in \cite{adh13:lics},
only substituting their blackbox for a separation oracle by the one we obtained
from Section \ref{sec:sep}. To avoid duplication of the parts that stay unchanged,
we refer to their work for the in-depth description of the algorithm, 
including the definition of the refinement procedure. This then concludes the 
proof for Lemma \ref{lem:optdefinable}.

Together with the result from Section \ref{sec:sep} that the weak separation 
oracle for the feasible region of an explicitly given SDP can be defined in \FPC,
this now almost establishes our main result of Theorem \ref{thm:sdpdefinable}. 
A small technicality still remains: We assume as a condition in Theorem 
\ref{thm:sdpdefinable} that the feasible region of the given SDP instance is 
bounded and non-empty, while the original reduction given in Theorem \ref{thm:GLSopt}
assumes the region to be bounded and full-dimensional.  However, by means of a 
simple preprocessing step, a non-empty region can be turned into a
full-dimensional one.

Assume we are given an SDP with a feasible region of 
$\mathcal{F}=\{X\in \QQ^{V\times V}\mid X\succeq 0, \innerprod{A_i}{X}\leq b_i, A_i\in \mathcal{A}, b_i\in b\}$,
and we want to find a $\delta$-close and $\delta$-maximal point of $\mathcal{F}_{\mathcal{A},b}$
with respect to some objective matrix $C$. We can then define an enlarged feasible region 
$\mathcal{F}':=\in \QQ^{V\times V}\mid (X+\frac{\epsilon}{\sqrt{|V|\norm{C}}} I)\succeq 0, \innerprod{A_i}{X}\leq b_i+\frac{\epsilon}{\norm{A_i}\norm{C}}\}$
for some $\epsilon > 0$. Note that $\mathcal{F}'$ is a non-empty, full-dimensional
convex set, where every point is at most $\epsilon$ far away from the original 
set $\mathcal{F}$.  Furthermore, it holds that $\max_{X\in\mathcal{F}'}\innerprod{C}{X}\leq \max_{X\in\mathcal{F}}+\epsilon\innerprod{C}{X}$.
Hence, any $\delta$-close and $\delta$-maximal point of $\mathcal{F}'$
is also a $\delta+\epsilon$-close and $\delta+\epsilon$-maximal point of 
$\mathcal{F}$. Consequently, by choosing $\epsilon$ sufficiently small,
we can simply perform the optimization over the full-dimensional set, 
which is covered by Theorem \ref{thm:GLSopt}.

This concludes the proof of Theorem \ref{thm:sdpdefinable}. Not that the conditions
on the feasible region of the definable SDP instances are readily satisfied 
for instance by those arising from finite-valued $\CSP$s: The variables only
range in $[0,1]$, and there always exists a feasible solution. In fact, any 
(even non-optimal) assignment in the $\VCSP$ gives rise to a feasible solution
of the 0--1 LP instance. 

%% file: lasserre.tex
We now apply the definability result on SDPs obtained in the previous
section to prove Lemma~\ref{lem:lasserre_lb}. 

The following proposition allows us to translate approximate solutions to
exact ones. It quantifies the quality of approximation needed so that we can obtain the exact
optimum of the original 0--1 problem by rounding an approximate optimum of its
Lasserre \SDP.

\begin{proposition}\label{prop:rounding}
	Let $I=(A,b,c)$ be a 0--1 linear program whose optimal solution is integral.
	Its feasible region is given by $\mathcal{K}=\{x\in\QQ^V\mid Ax\geq b\}\cap\{0,1\}^V$
	with an objective vector $c\in\QQ^V$. Furthermore let $\Las^\pi_t(\mathcal{K})=\mathcal{K}^*$
	for some $t$, and let $s\in\QQ$ be the value of a $1/(4\max\{1,\norm{c}\})$-close
	and $1/(4\max\{1,\norm{c}\})$-maximal
	solution to $\Las_t(\mathcal{K})$ under the objective $c$. Then, by rounding
	$s$, we obtain the exact optimal value for $I$.
\end{proposition}
\begin{proof}
	Let $s^*$ be the exact optimal value of $\Las_t(\mathcal{K})$, and by assumption,
	also the optimal value for $I$. We argue that $|s-s^*|\leq 1/4$, and hence
	rounding $s$ yields $s^*$, since $s^*\in\ZZ$.
	
	First, note that the condition that $s$ is $1/(4\max\{1,\norm{c}\})$-maximal
	means that $s+1/(4\max\{1,\norm{c}\}) \geq \max_{x\in\mathcal{K}}\innerprod{c}{x}=s^*$.
	Hence, we have the lower bound $s\geq s^*-1/4$.
	
	The other direction follows from the fact that $s$ is the value of a
	close solution, say, $s=\innerprod{c}{y}$ 
	for some $y\in\QQ^V$. Since $y$ is $1/(4\max\{1,\norm{c}\})$-close
	to $\mathcal{K}^*$, it can be decomposed into $y=x+e$ where $x\in\mathcal{K}^*$
	and $\norm{e}\leq 1/(4\max\{1,\norm{c}\})$. The value of $y$ is then
	bounded by $s=\innerprod{c}{y}\leq \max_{x\in\mathcal{K}}\innerprod{c}{x}+\innerprod{c}{e}\leq s^* + 1/4$.
\end{proof}

Next, we show that it is possible in \FPC to define the $t$-th level of the
Lasserre hierarchy for any explicitly given 0--1 program using only $O(t)$
many variables.

\begin{lemma}\label{lem:lp_to_sdp}
	There is an \FPC-interpretation from $\tau_{\LP}$ to $\tau_{\SDP}$ that
	for a given 0--1 linear program expresses the $t$-th level Lasserre hierarchy,
	using at most $O(t)$ many variables.
\end{lemma}
\begin{proof}
	Let an instance $I$ of a 0--1 program be given by
	a matrix $A\in\QQ^{U\times V}$, and vectors $b\in\QQ^U, c\in\QQ^V$. In order to
	show that we can define the $t$-th level Lasserre relaxation from $I$, it
	suffices to show that we can define the matrices $M_t(y)$ and $S^u_t(y)$ from $I$
	for any $y\in\QQ^{\powerset_{2t+1}(V)}, u\in U$. In particular, we
	represent $M_t(y)$ as a sequence of matrices $(\hat{M}_{t,q})_{q\in Q}$ where
	$Q:=\powerset_{2t+1}(V)\cup\{0\}$, such that 
	$M_t(y)=\hat{M}_{t,0}+\sum_{q\in Q\backslash\{0\}}y_q \hat{M}_{t,q}$, and show
	that these matrices are definable. We represent $S^u_t(y)$ in an analogous way
	as $S^u_t(y)=\hat{S}^u_{t,0}+\sum_{q\in Q\backslash\{0\}}y_q \hat{S}^u_{t,q}$.
	Hence, here it suffices to argue that the matrices 
	$(\hat{M}_{t,q})_{q\in Q}$ and $(\hat{S}^u_{t,q})_{q\in Q}$ are definable
	from $I$ within $\FPC$.
	
	Observe that for the $t$-th level Lasserre relaxation the matrices $M_t(y)$ and
	$S^u_t(y)$ are indexed by powersets $\powerset_t(V)$, and the feasible region
	itself lies in a vectorspace indexed by $\powerset_{2t+1}(V)$.
	As we need these index sets in our interpretation, we first describe 
	how to define the powersets $\powerset_k(V)$ for some fixed $k$. Namely, we 
	encode a set $S\in\powerset_k(V)$ by a $k$-ary tuple $T\in (V\cup\{0\})^k$
	that contains each of the elements in $S$ once, and where the symbol $0$ fills 
	the rest of the positions. As there are up to $k!$ many tuples encoding
	the same set, we additionally define an equivalence relation $\tilde{=}_k$
	on $k$-tuples that identifies two tuples if they are just permutations of 
	each other. This can be defined by the following first order formulas.
	\[\delta_{\powerset_k}(x_1,\ldots,x_k):=\bigwedge_{i\in [k]}(x_i=0 \vee x_i\in V) 
		\bigwedge_{i,j\in [k]}(x_i=0 \vee x_i\neq x_j),\]
	\[\tilde{=}_k(x_1,\ldots,x_k,y_1,\ldots,y_k):=\bigvee_{\pi\in Sym(k)}\bigwedge_{i\in [k]}x_i=y_{\pi(i)},\]
	where $\delta_{\powerset_k}$ defines the set of tuples encoding some element
	in $\powerset_k(V)$, and $\tilde{=}_k$ defines a binary equivalence relation between
	those tuples.  We use $Sym(k)$ to denote the set of permutations on $[k]$.
	From these definitions it is not hard to define basic set operations on the
	elements of $\powerset_k(V)$. For instance, we can define a $4k$-ary relation
	$union_k$ that encodes the union of two sets $S,T\in\powerset_k(V)$.
	\begin{align*}
		union_k(x,s,t)=
		&\bigwedge_{i\in[2k]}\left(x_i=0 \bigvee_{j\in [k]} x_i=s_j \vee x_i=t_j\right)\\
		& \bigwedge_{i\in [k]}\bigvee_{j\in[2k]}s_i = x_j
		\bigwedge_{i\in [k]}\bigvee_{j\in[2k]}t_i = x_j,
	\end{align*}
	where $x\in\delta_{\powerset_{2k}}$, and $s,t\in\delta_{\powerset_{k}}$. Since
	$union_k(x,s,t)$ simply encodes $(x=s\cup t)$, we continue using the latter
	more familiar notation for set operations. One point to note here is
	that all formulas so far are all defined using $O(k)$ many variables.
	
	Now we can turn to the definition of the matrices 
	$(\hat{M}_{t,q})_{q\in Q}$ and $(\hat{S}^u_{t,q})_{q\in Q}$. 
	Each matrix $\hat{M}_{t,q}$ and $\hat{S}^u_{t,q}$ is indexed over the set
	$\delta_{\powerset_t}\times\delta_{\powerset_t}$. For $x,y\in\delta_{\powerset_t}$
	and $q\in\delta_{\powerset_{2t+1}}$, their entries are given by

	\[\hat{M}_t(q,x,y)=\begin{cases}
		1 &\mbox{if } x\cup y = q \\
		0 &\mbox{otherwise},
	\end{cases}\]
	and
	\[\hat{S}^u_t(q,x,y)=\begin{cases}
		A_{u,v} &\mbox{if } \exists v\in V: x\cup y\cup\{v\} = q \\
		-b_u &\mbox{if } x\cup y = q \\
		0 &\mbox{otherwise}.
	\end{cases}\]
	By the above expressions $(\hat{M}_{t,q})_{q\in Q}$ and $(\hat{S}^u_{t,q})_{q\in Q}$ 
	are definable in \FPC  using only $O(t)$ variables. From this, we obtain
	the full SDP in inequality standard form by merging
	the constraints $M_t(y)\succeq 0$ and $S^u_t(y)\succeq 0$ 
	into one single constraint of the form $Z\succeq 0$. 
	
\end{proof}
	
\begin{lemma}\label{lem:csp_to_lp}
	Let $D$ be a domain, and $\Gamma$ a finite-valued constraint language. 
	There is an \FPC-interpretation of constant width from
	$\tau_{\Gamma}$ to $\tau_{LP}$ that defines for a given instance $I$ of 
	$\VCSP(D,\Gamma)$ its corresponding 0--1 linear program.
\end{lemma}
\begin{proof}
	The 0--1 program that encodes a $\VCSP$ instance is given in Section~\ref{sec:csps}.
	It has been shown in \cite{DW15} that this LP is definable in \FPC. The
	construction there also only uses a constant number of variables.
\end{proof}

Finally, we are now ready to prove Lemma~\ref{lem:lasserre_lb}.

\begin{proof}[Proof of Lemma \ref{lem:lasserre_lb}]
	For the proof we fix a domain $D$ and 
	a finite-valued constraint language $\Gamma$. For better legibility, we 
	write $L_\Gamma$
	for $L_{\VCSP(D,\Gamma)}$, $\nu_{\Gamma}$ for $\nu_{\VCSP(D,\Gamma)}$,
	and $\tau_\Gamma$ for the vocabulary $\tau_{\VCSP(\Gamma)}$.
	
	The proof idea is as follows. The argument is by contradiction. 
	Suppose that $L_\Gamma(n)\in o(\nu_\Gamma(n))$. However, by 
	composing the interpretations from Lemmas \ref{lem:csp_to_lp} and
	\ref{lem:lp_to_sdp} and Theorem \ref{thm:sdpdefinable} we can
	define a formula $\phi$ that decides membership for the decision version of 
	$\VCSP(D,\Gamma)$ for 
	instances of size $n$ using only $o(\nu_\Gamma(n))$ many variables, which
	violates the assumed counting width bound of $\nu(n)$.

	To be more precise, let $\Theta_t$ be the composition of the interpretations from 
	Lemmas \ref{lem:csp_to_lp} and \ref{lem:lp_to_sdp}. That is, $\Theta_t$ is an
	interpretation of $\tau_{\SDP}$ in $\tau_\Gamma$ that defines 
	for a given $\VCSP$ instance $I=(V,C,w)$ the SDP of 
	the $t$-th level of the Lasserre relaxation of the corresponding
	0--1 linear program. Note that $\Theta_t$ is of width $O(t)$.
	
	Note that the 0--1 linear programs corresponding to $\VCSP$ instances 
	are always feasible, bounded in the 0--1 hypercube, and their optimum is 
	always integral.
	
	Suppose now $L_\Gamma(n)\in o(\nu_\Gamma(n))$, i.e.\ every instance 
	$I=(V,C,w)$ of 
	$\VCSP(D,\Gamma)$ could be captured by some Lasserre relaxation of level
	$t\in o(\nu_\Gamma(|V|))$. Hence, $\Theta_t(I)$ defines a Lasserre relaxation whose
	optimal value is exactly the optimal value to $I$.
		
	Then, by Theorem \ref{thm:sdpdefinable} there is an interpretation $\Sigma$ of 
	$\vocvec$ in $\tau_{\SDP}\dunion\tau_{\QQ}$ that defines
	$\delta$-close and $\delta$-maximal solutions to $\Theta_t(I)$. 
	Using Proposition \ref{prop:rounding}, setting
	$\delta$ as $\delta=1/4|C|$ allows us to obtain the
	exact optimal value for $\Theta_t(I)$ (and equivalently, for $I$)
	by means of rounding. Both
	defining the value for $\delta$ as well as the rounding can be done in \FPC.
	Hence composing $\Theta_t$ and $\Sigma$, we obtain an interpretation $\Phi$ of
	width $O(t)$ that defines for a given instance of $\VCSP(D,\Gamma)$ 
	its optimal value.
	
	Finally, using $\Phi$ it is not difficult to construct a \FPC-formula $\phi$
	using at most $O(t)$ many variables that decides membership for the
	decision version of $\VCSP(D,\Gamma)$:
	For an instance $(I,t)$ we simply compare $\Phi(I)$ to $t$. Since we assumed
	$t\in o(\nu_\Gamma(|V|))$, $\phi$ also uses only $o(\nu_\Gamma(|V|))$ many
	variables. This is a contradiction to the definition of $\nu_\Gamma$.
\end{proof}

%% file: conclusion.tex
We have established a dichotomy result, showing that \emph{every} finite-valued CSP that is not solvable by its basic linear programming relaxation (and this includes all constraint maximization problems that are known to be NP-hard) requires a linear number of levels of the Lasserre hierarchy to solve exactly.  Such linear lower bounds on the number of levels of the Lasserre hierarchy were known previously for specific CSPs.  Our result shows that these are part of a sweepingly general pattern.  This is established by considering the definability of semidefinite programs in logic, and using  a measure of logical complexity, that we call \emph{counting width}, to classify CSPs.  This suggests some directions for further investigation.

A central motivating interest in semidefinite programming in general and Lasserre hierarchies in particular comes from their use in approximation algorithms.  It would be interesting to extend our methods to show lower bounds on the levels required to approximate a solution, as well as to obtain exact solutions.  A potential direction is to define a measure of counting width, not just for a class of structures $\mathcal{C}$ but based on the number of variables to separate two classes $\mathcal{C}_1$ and $\mathcal{C}_2$.  We could then seek to establish lower bounds on these numbers where $\mathcal{C}_1$ is a collection of instances of a $\VCSP$ with high optimum values and $\mathcal{C}_2$ contains only instances with low optima.  This would show that instances with high optima cannot be separated from those with low optima by means of a small number of levels the Lasserre hierarchy.

Definability in $\FPC$ is closely linked to \emph{symmetric} computation (see~\cite{AD16,Daw15}).  In other words, algorithms that can be translated to this logic are symmetric in a precise sense.  This suggests that many of our best approximation algorithms for constraint satisfaction, such as the Lasserre semidefinite programs are encountering a ``symmetry barrier''.  Breaking through this barrier, and coming up with algorithms that break symmetries, may be crucial to more effective approximation algorithms.